\documentclass[11pt,hidelinks]{article}
\usepackage{amsmath,amsfonts}
\usepackage{enumerate}
\usepackage{natbib}
\usepackage{url} 
\usepackage{algpseudocode}
\usepackage{soul} 
\usepackage{amsthm}
\usepackage{bbm}
\usepackage[pdftex]{graphicx}
\usepackage{indentfirst}
\usepackage{verbatim}
\usepackage{setspace}
\usepackage{mathtools}
\usepackage[capposition=top]{floatrow}
\usepackage{booktabs}
\usepackage{enumitem}
\usepackage[flushleft]{threeparttable}
\usepackage{comment}
\usepackage{color}
\usepackage{bm}
\usepackage{subcaption}
\usepackage{changepage}
\usepackage{hyperref}
\usepackage[toc,title,page]{appendix}
\usepackage[dvipsnames]{xcolor}
\usepackage{ longtable, xpatch,  dcolumn,booktabs,adjustbox}

\usepackage{mathtools}

\newcommand{\E}{\mathbb{E}}
\newcommand{\Var}{\text{Var}}
\newcommand{\PP}{\mathbb{P}}

\newcommand{\Lp}{\left(}
\newcommand{\Rp}{\right)}
\newcommand{\Ls}{\left[}
\newcommand{\Rs}{\right]}
\newcommand{\Lv}{\left\|}
\newcommand{\Rv}{\right\|}
\newcommand{\La}{\left|}
\newcommand{\Ra}{\right|}

\usepackage{amsmath}
\usepackage{chngcntr}

\makeatletter
\newcommand{\Biggg}{\bBigg@{3.5}}

\newcommand{\norm}[1]{\left\lVert#1\right\rVert}

\newcommand{\indep}{\perp \!\!\! \perp}

\definecolor{Blue}{rgb}{0,0,1}
\definecolor{Grey}{rgb}{.5,.5,.5}

\usepackage{multirow, rotating}
\usepackage{pst-node}
\usepackage{inconsolata}
\usepackage[T1]{fontenc}

\newtheorem{lemma}{Lemma}
\newtheorem{assumption}{Assumption}

\newtheorem{proposition}{Proposition}

\newtheorem{theorem}{Theorem}
\newtheorem{corollary}{Corollary}
\usepackage{natbib}

\usepackage{amsthm}
\theoremstyle{definition}
\newtheorem{example}{Example}
\newenvironment{examplecontinued}[1]
  {\renewcommand{\theexample}{#1 continued}\example}
  {\endexample}

\DeclareMathOperator*{\argmin}{arg\,min}
\usepackage{geometry} 
\usepackage{hyperref}
\usepackage{amsmath}
\usepackage{amssymb}
\usepackage{booktabs}
\usepackage{enumitem}
\hypersetup{
  colorlinks   = true, 
  urlcolor     = blue, 
  linkcolor    = blue, 
  citecolor   = blue 
}
\setcitestyle{open={(},close={)}}

\newcommand*{\thisdraft}{This draft: May 11th, 2025}

\begin{document}

\title{Automatic Doubly Robust Forests}

\author{
    Zhaomeng Chen\thanks{These authors contributed equally and are listed alphabetically.} \\
    Department of Statistics \\
    Stanford University \\
    \texttt{zc313@stanford.edu}
    \and
    Junting Duan\footnotemark[1] \\
    Department of Management Science and Engineering \\
    Stanford University \\
    \texttt{duanjt@stanford.edu}
    \and
    Victor Chernozhukov \\
    Department of Economics and Center for Statistics and Data Science \\
    Massachusetts Institute of Technology \\
    \texttt{vchern@mit.edu}
    \and
    Vasilis Syrgkanis\thanks{Vasilis Syrgkanis was supported by NSF Award IIS-2337916.} \\
    Department of Management Science and Engineering \\
    Stanford University \\
    \texttt{vsyrgk@stanford.edu}
}

\date{\thisdraft}

\begin{titlepage}
\maketitle
\begin{abstract}

This paper proposes the automatic Doubly Robust Random Forest (DRRF) algorithm for estimating the conditional expectation of a moment functional in the presence of high-dimensional nuisance functions. DRRF extends the automatic debiasing framework based on the Riesz representer to the conditional setting and enables nonparametric, forest-based estimation \citep{grf,orf}. In contrast to existing methods, DRRF does not require prior knowledge of the form of the debiasing term or impose restrictive parametric or semi-parametric assumptions on the target quantity. Additionally, it is computationally efficient in making predictions at multiple query points. We establish consistency and asymptotic normality results for the DRRF estimator under general assumptions, allowing for the construction of valid confidence intervals. Through extensive simulations in heterogeneous treatment effect (HTE) estimation, we demonstrate the superior performance of DRRF over benchmark approaches in terms of estimation accuracy, robustness, and computational efficiency.

\vspace{1cm}

\end{abstract}

\end{titlepage}

\section{Introduction} \label{sec: intro}
A wide range of problems in social and life sciences can be framed as estimating the conditional expectation of a moment functional that depends on an unknown nuisance regression function:
\begin{equation}  \label{equ:problem}
    \E[m(Z;g_0(x,W))\mid X=x]\coloneqq \theta_0(x).
\end{equation}
Here, $m(z;g)$ is a known moment functional, and $g_0(x,w)=\E[Y\mid W=w,X=x]$ is an unknown nuisance regression function. The observation $Z=(Y,X,W)$ consists of an outcome variable $Y$, target covariates $X$, and additional variables $W$. A key example is heterogeneous treatment effect (HTE) estimation in causal inference, where the goal is to understand how a treatment $D$ affects the outcome $Y$ across different sub-populations or covariate groups defined by $X$. In HTE estimation, $W=(D,\tilde{X})$ includes the treatment variable $D$ and additional confounders $\tilde{X}$ that need to be controlled for. HTE estimation has crucial applications across various fields. For example, in personalized medicine, it helps identify effective treatments tailored to individual patient characteristics. In marketing, it guides customer segmentation and enables targeted interventions.

A primary challenge in estimating \eqref{equ:problem} arises from the high dimensionality of the nuisance regression function. In many applications, it is desirable for $g_0$ to depend on high-dimensional variables $W$. For example, in HTE estimation, controlling for a large number of confounders $\tilde{X}$ is crucial for the validity of the causal identification argument. To achieve high-quality predictions in high-dimensional settings, modern machine learning (ML) methods with regularization and model selection are frequently used. However, these methods often introduce non-negligible bias due to the bias-variance trade-off inherent in optimizing mean squared error. This bias poses a significant challenge for inference, as it can lead to poor coverage of confidence intervals for the target estimand \citep{chernozhukov2018double,chernozhukov2022locally}.

Much of the prior literature on non-parametric estimation of $\theta_0(x)$ assumes low-dimensional nuisance parameterizations \citep{wager2018estimation,grf}. While this simplifies the estimation process, it limits applicability in real-world applications where a large number of additional variables $W$ are present. 
On the other hand, to accommodate high-dimensional nuisance functions, some existing works construct Neyman orthogonal moment functionals, which eliminate the first-order effect of nuisance estimation on the expected moment functional. However, these works typically impose parametric assumptions on the target function $\theta_0(x)$. For example, \cite{victor17} and \cite{chernozhukov2022locally} focus on estimating the average treatment effect (ATE), where $\theta_0(x)$ is assumed to be constant across $x$. \cite{chernozhukov2017orthogonal} and \cite{chernozhukov2018plug} consider cases where $\theta_0(x)$ follows a linear specification, $\theta_0(x)=\theta_0^\top \phi(x)$, with $\theta_0$ assumed to be sparse and $\phi(x)$ a known feature mapping. Estimating a non-parametric $\theta_0(x)$ would offer more flexibility but is inherently more challenging than its parametric counterpart.

In this paper, we propose the Doubly Robust Random Forest (DRRF) algorithm to address these challenges of estimating conditional moment functionals. Our method extends the automatic debiased machine learning (Auto-DML) framework, which targets the conditional expectation of average moment functionals, to the more general setting of conditional moment functionals. To achieve this, we generalize the notion of the Riesz representer to a conditional Riesz representer and construct a doubly robust moment by incorporating a debiasing term based on this conditional representer. This debiasing mitigates the regularization bias inherent in plug-in machine learning methods and thereby accommodates high-dimensional nuisance functions. The estimation of the conditional Riesz representer is automatic, in the sense that it does not require prior knowledge of its form, and it can align with the estimation of the regression function \( g_0 \). We show that, under suitable conditions, the conditional Riesz representer can be projected onto the same functional space as \( g_0 \). For example, if \( g_0 \) is locally linear, we may, without loss of generality, restrict the conditional Riesz representer to be locally linear as well. We provide an efficient algorithm for estimating the conditional Riesz representer and constructing the debiasing term based on the forest lasso method of \citet{orf}.

Our method is a two-stage procedure, similar to the Orthogonal Random Forest (ORF) algorithm proposed by \citet{orf}. In the first stage, we estimate the nuisance functions and, consequently, the conditional moments. In the second stage, we compute similarity weights between each data point and the query point \( x \). The final estimator is constructed as a weighted average of the estimated conditional moments using a data-splitting approach. We estimate the weights following the strategy of \citet{orf}, which enables nonparametric, forest-based estimation of the conditional functional \( \theta_0(x) \). However, our method significantly deviates from ORF in the manner it estimates the nuisance functions, making DRRF much more practical for large scale deployment when predicting $\theta_0(x)$ at a large number of query points $x_1,\ldots,x_m$. Specifically, the ORF algorithm requires re-estimating the ML-based nuisances for every single query point $x_i$, each time using different sample weights on the training data that are tailored to the query point. However, our algorithm estimates the nuisance functions only once for each of the training data points and does not re-estimate them for each query point after the training phase, as required by ORF. Since the estimation of nuisances is typically the most computationally intensive step, our approach substantially reduces computation when handling a large number of queries. This alternative nuisance estimation method also provides side benefits with respect to interpretability, which is important for most downstream decision making applications of the heterogeneous model $\theta_0(x)$. The final model $\hat\theta(\cdot)$ of the DRRF method is equivalent in representation to a standard regression forest, where the leaf nodes of each tree contain local estimates and the prediction of each tree at any query point $x$ is the value of the leaf node in which $x$ falls, and the prediction of the forest is simply the average of the tree predictions. Hence, standard interpretability tools for regression forests, such as SHAP \citep{lundberg2017unified}, can be used off-the-shelf to interpret the model and provide insights into which features are important, on aggregate, and how each feature contributes to each query result $\theta_0(x_i)$. This is not the case for the ORF method, since evaluating $\theta_0(x)$ for any target $x$ involves complex nuisance re-estimation and is not a simple ``look-up and average'' calculation.

We show that our estimator is consistent and asymptotically normal, which allows for the construction of asymptotically valid confidence intervals. By incorporating the debiasing term, DRRF derives a doubly robust moment functional that relaxes the convergence rate requirements for the nuisance regression functions needed for consistency. Under appropriate regularity conditions, we show that our estimator achieves the same convergence rate as an oracle estimator with known nuisance functions, provided that the product of the estimation errors for the nuisance regression function and the conditional Riesz representer remains small. This property makes DRRF particularly advantageous in settings with high-dimensional and complex nuisance functions where accurate estimation of the conditional Riesz representer is feasible. Under additional conditions, we show that our estimator is asymptotically normal, enabling the construction of confidence intervals using bootstrap methods.

Variants of our DRRF estimation algorithm have already been widely used in practice. For instance, the ForestDRLearner method, implemented in the widely adopted open-source package EconML \citep{econml}, is essentially the non-automatic analogue of our DRRF estimator for the HTE problem. However, no prior work has established the asymptotic normality of this method. The closest related prior work is the doubly robust variant of the ORF method, which however constructs the nuisance functions in a different manner, as described above—a distinction crucial to their proof. Hence, our asymptotic normality theorem not only provides a formal justification of these widely used methods but also enhances and generalizes them by incorporating automatic debiasing techniques. Automatic debiasing has also been shown to be empirically less sensitive to extreme propensity problems \citep{riesznet,RieszSharma}, as it directly targets learning the Riesz function $\alpha_0$ rather than first estimating propensity models and then constructing the Riesz function through explicit division.

Our work contributes to two streams of recent literature at the intersection of causal inference and machine learning. 
First, our work builds on the literature on the estimation and inference of treatment effects that are robust to errors introduced by sophisticated machine learning methods. \cite{victor17} propose the double machine learning (double ML) method, which constructs a Neyman orthogonal moment by residualizing the effects of high-dimensional covariates on both treatments and outcomes. The treatment effect is then estimated through a linear regression between the residualized variables. 
\cite{chernozhukov2022automatic} introduce the automatic debiased machine learning (Auto-DML) framework based on the (unconditional) Riesz representer of the score function, which ensures Neyman orthogonality and achieves double robustness for the target moment functional. \cite{riesznet} further develop automatic procedures to learn the Riesz representation using neural networks and random forests. Building on this line of work, we extend the Auto-DML framework to conditional setting. We introduce the conditional Riesz representer and use it to construct debiased estimators for conditional functionals. Our conditional Riesz represneter can also be automatic estimated, resulting in a fully automatic debiasing procedure for conditional targets.

Second, our work contributes to the literature on machine learning-based estimation of conditional moment functions. \cite{grf} introduce Generalized Random Forest (GRF), a flexible tree-based method for solving general conditional moment equations. GRF builds trees with partitions designed to capture heterogeneity in the target quantity. However, it cannot directly perform high-dimensional local nuisance estimations, so in practice, it is often combined with global nuisance function estimators in high-dimensional settings, which may not align well with the conditional estimand. \cite{orf} extend GRF to the Orthogonal Random Forest (ORF) algorithm by combining the forest-based algorithm with Neyman orthogonal moments. While ORF allows for local fitting of high-dimensional nuisance functions around the target feature, it has several limitations. First, ORF employs double ML to construct Neyman orthogonal moments. Effective residualization in double ML typically requires well-specified models for the data generating process. For example, when estimating HTE, accurately modeling both the outcome and treatment variables is often crucial for achieving strong performance. Second, ORF refits nuisance functions for each target prediction point $x$, leading to high computational costs when the number of query points is large. 
In contrast, our DRRF method automatically constructs a doubly robust moment functional for estimating $\theta_0(x)$ without relying on any parametric or semi-parametric assumptions about the treatment variable. DRRF also facilitates efficient estimation of $\theta_0(x)$ across multiple query points.
Additionally, compared to other doubly robust estimators that rely on the inverse of propensity score estimates, such as ForestDRLearner and the one proposed by \cite{zimmert2019nonparametric}, our approach directly estimates the debiasing term, ensuring stability even when propensity scores are near boundary values.

We conduct extensive simulations to demonstrate the superior performance of the DRRF algorithm in HTE estimation over benchmark approaches, including variants of ORF and GRF. Our results highlight three main advantages of DRRF. First, like ORF, DRRF performs local fitting of nuisance functions and generally outperforms GRF variants that rely on global fitting. Second, DRRF performs automatic debiasing and does not require a well-specified model for the treatment variable. In contrast, ORF may suffer from significant bias or variance if the treatment model is misspecified, whereas DRRF remains stable and provides accurate estimates. Lastly, DRRF is significantly more computationally efficient than ORF, with much shorter evaluation runtime when estimating at multiple target query points.

The rest of the paper is organized as follows. In Section \ref{sec: dr_cme}, we introduce our doubly robust conditional moment estimation procedure. In particular, we define the general framework in Section \ref{subsec: dr estimation} and propose the Doubly Robust Random Forest algorithm in Section \ref{subsec: drrf} as an efficient implementation of the general framework. In Section \ref{sec: theory}, we provide theoretical results that establish the consistency and asymptotic normality of our estimator under appropriate conditions. Finally, in Section \ref{sec:simulations}, we conduct comprehensive synthetic data simulations to demonstrate the advantages of our method over benchmark approaches.

\section{Problem Statement} \label{subsec: ps}

Consider a dataset with $2n$ i.i.d. observations $Z_i=(Y_i,X_i,W_i)$, where each observation $Z_i$ contains a scalar outcome of interest $Y_i$, target covariates $X_i$, and additional variables $W_i$. This paper provides an automatic, doubly robust, and computationally efficient approach for estimating the conditional expectation of a moment functional
\begin{align} \label{eq: target}
    \theta_0(x) & = \E[m(Z;g_0(x,W))\mid X=x],
\end{align}
evaluated at multiple target points $x=x_1,\ldots,x_m$. Here, $g_0(x,w)$ is an unknown nuisance regression function defined by $g_0(x,w) \coloneqq \E[Y\mid W=w,X=x]$, and $m(z;g)$ is a known moment functional that depends on a data observation $z$ and a possible regression function $g$. We assume that the target feature space $\mathcal{X}=[0,1]^d$ has a fixed dimension $d$, and that $X$ has a density bounded away from both zero and infinity. We allow the dimension of $W$ to be large and grow with $n$, positioning the problem in a high-dimensional setting where the nuisance regression function can depend on a significant number of variables.

To illustrate the general target in \eqref{eq: target}, we consider estimation problems related to heterogeneous treatment effects (HTE) in Examples \ref{ex:cate}-\ref{ex:cipe}. These examples demonstrate the broad applicability of our framework to moment functionals in causal inference. They also highlight the importance of including a large number of covariates in the nuisance regression function $g_0$, as doing so enhances the plausibility of the unconfoundedness assumption for treatment effect identification.

\begin{example}[Conditional Average Treatment Effect] \label{ex:cate} Suppose that $D$ is a binary treatment variable, and $Y(d)$ is the potential outcome under treatment $d$. Our objective is to estimate the average treatment effect conditional on the target covariates $X$: 
$$\theta_0(x) = \E[Y(1)-Y(0)\mid X=x],$$
at multiple values $x$. This problem is practically important for capturing fine-grained treatment effects that account for individual heterogeneity across a diverse population.

Under the unconfoundedness assumption that $D \perp\!\!\!\perp Y(d) \mid X,\tilde X$, we have
\[
\theta_0(x) = \E\Ls g_0(x,(1,\tilde X)) - g_0(x,(0,\tilde X)) \mid X=x\Rs.
\]
In this example, $W =(D,\tilde X)$ includes the treatment indicator $D$ and additional control covariates $\tilde X$. We define the moment function in \eqref{eq: target} as $m(z; g) = g(x, (1, \tilde{x})) - g(x, (0, \tilde{x})),$ so that the problem falls within the general framework described above.

\end{example}

\begin{example}[Conditional Average Marginal Effect] \label{ex:came} 
With continuous treatment variable $D$, the conditional average marginal effect is defined by \[
\theta_0(x) = \E\Ls \left.\frac{\partial Y(D)}{\partial d} \ \right| \ X=x\Rs.
\]
This quantity captures the average rate of change in the potential outcome with respect to the treatment, for an individual with features $X = x$.

Under the unconfoundedness assumption that $D \perp\!\!\!\perp Y(d) \mid X,\tilde X$,
\[
\theta_0(x) = \E\Ls \left.\frac{\partial g_0(x,(D,\tilde X))}{\partial d} \ \right| \ X=x\Rs.
\]
Therefore, this fits into the general problem \eqref{eq: target} by taking $W= (D,\tilde X)$ and $m(z;g)=\partial g(x,(d,\tilde x))/\partial d$. 
\end{example}

\begin{example}[Conditional Incremental Policy Effect] \label{ex:cipe} Using the same notation as in Example~\ref{ex:came}, the incremental policy effect for a policy $\pi(x, \tilde{x}) \in [-1, 1]$ is defined by 

$$\theta_0(x) = \E\Ls\left.\left. \frac{d}{d\nu}\right|_{\nu=0} Y(D+\nu\pi(x,\tilde{X})) \right| \ X=x\Rs.$$ The incremental policy effect captures the average instantaneous change in the outcome $Y$ induced by shifting the treatment in the direction of the policy $\pi(x, \tilde{x})$, conditional on $X = x$.

Under the unconfoundedness assumption that $D \perp\!\!\!\perp Y(d) \mid X,\tilde X$, we have 

\[
\theta_0(x) = \mathbb{E} \left[ \left. \pi(x,\tilde{X}) \frac{\partial g_0(x,(D,\tilde X))}{\partial d} \, \right| \, X = x \right].
\]
Therefore, we define $W = (D,\tilde X)$ and take the moment function as $m(z;g)=\pi(x,\tilde{x})\partial g(x,(d,\tilde x))/\partial d$. 

\end{example}

\paragraph{Notations.} In the following, we define $\norm{g}_x = (\E[g(x,W)^2\mid X=x])^{1/2}$ and let $\mathcal{G}_x$ be the set of functions such that $\norm{g}_x<\infty$. We use $f(n)=\omega(g(n))$ to denote that $f$ asymptotically dominates $g$, i.e., $f(n)/g(n)\rightarrow \infty$ as $n\rightarrow \infty$, and $f(n)=\Theta(g(n))$ denotes that $f$ is asymptotically bounded by $g$ both from above and below, i.e., $f(n)=O(g(n))$ and $g(n)=O(f(n))$.

\section{Automatic Doubly Robust Conditional Moment Estimation} \label{sec: dr_cme}

\subsection{Conditional Riesz Representer and Debiased Moment Functional}

In this section, we show how to adapt the automatic debiasing method using the Riesz representer \citep{chernozhukov2022automatic} to estimate the \textit{conditional} target $\theta_0(x)$. To apply the Riesz representation theorem, we assume that, for any $x \in \mathcal{X}$, the conditional expected moment $\E[m(Z;g(x,W))\mid X=x]$ is a continuous linear functional of $g$ with domain $\mathcal{G}_x$. Note that the conditional function space $\mathcal{G}_x$ is the $L_2$ space under the law $\mathcal{L}([X,W]\mid X=x)$, and specifically, it is a Hilbert space with inner product defined by $$\langle g_1,g_2\rangle_x = \E[g_1(X,W)g_2(X,W)|X=x]. $$

By the Riesz representation theorem, there exists a unique function $\alpha_0\in \mathcal{G}_x$, such that for any $g\in \mathcal{G}_x$,
\begin{align}
    \E[m(Z;g(x,W))\mid X=x] = \E[\alpha_0(x,W)g(x,W)\mid X=x].
\end{align}

We refer to $\alpha_0$ as the \textit{conditional Riesz representer} (cRR) to highlight its distinction from the standard Riesz representer used for global moment functionals: the cRR depends explicitly on the conditioning variable $X = x$. 
To first build some intuition for the cRR $\alpha_0$ and its role in our framework, we illustrate it in the context of CATE estimation in Example \ref{ex:cate}.

\begin{examplecontinued}{1}[Conditional Average Treatment Effect]
For any function $g\in \mathcal{G}_x$,
\begin{align*}
    &\E[g(x,(1,\tilde X)) - g(x,(0,\tilde X)) \mid X=x] \\
    = \  & \E\left[\left.\left( \frac{D}{\PP(D=1\mid \tilde X,X=x)} - \frac{1-D}{\PP(D=0\mid \tilde X,X=x)} \right) g(x,W)\ \right| X=x \right].
\end{align*}
Therefore, under mild conditions,\footnote{For example, it suffices to assume that $\PP(D=1\mid \tilde X,X=x)$ is bounded away from 0 and 1, so that the conditional Riesz representer $\alpha_0(x,W)$ defined by the propensity scores satisfies $\alpha_0\in \mathcal{G}_x$.} the conditional Riesz representer $\alpha_0\in \mathcal{G}_x$ is given by $$\alpha_0(x,W) = \frac{D}{\PP(D=1\mid \tilde X,X=x)} - \frac{1-D}{\PP(D=0\mid \tilde X,X=x)}.$$
\end{examplecontinued}

Using the conditional Riesz representer, we construct a debiased moment functional $\psi(Z; \alpha_0, g)$ for the original moment functional $m(Z; g)$. This construction is analogous to the use of unconditional Riesz representer in the debiased machine learning literature for estimating unconditional estimands. Specifically, we define  
\[
\psi(Z; \alpha_0, g) \coloneqq m(Z; g(X, W)) + \alpha_0(X, W)(Y - g(X, W)).
\]

It is straightforward to show that the debiased moment functional $\psi$ satisfies $\E[\psi(Z;\alpha_0,g_0)\mid X=x]=\theta_0(x)$ by conditioning on $X$ and $W$ and applying the law of iterated expectations. Therefore, we can estimate $\theta_0(x)$ through a plug-in weighted average of $\psi(Z;\hat{\alpha},\hat{g})$, using estimators for both $g_0$ and $\alpha_0$. Importantly, the debiased moment functional $\psi(Z;\alpha,g)$ is locally Neyman orthogonal, meaning that the regression estimator $\hat{g}$ and the cRR estimator $\hat{\alpha}$ have no first-order effect on the moment at the target point $x$. Furthermore, $\psi(Z;\alpha,g)$ is doubly robust for the estimation of $\theta_0(x)$, which allows our estimator to achieve a fast convergence rate under mild consistency requirement for the nuisance estimators $\hat{g}$ and $\hat{\alpha}$. We will elaborate on this property in our theoretical results in Section \ref{sec: theory}. In Propositions \ref{prop: ortho} and \ref{prop: dr} below, we formally state the local Neyman orthogonality and double robustness properties of the debiased moment functional $\psi(Z;{\alpha},{g})$.

\begin{proposition}[Local Orthogonality] \label{prop: ortho} 
    For any fixed estimators $\hat{g}$ for the regression function and $\hat{\alpha}$ for the conditional Riesz representer, the local gradients of $\psi$ with respect to $g$ and $\alpha$ satisfy
    \begin{equation}
        \begin{aligned}
        &\E\Ls \nabla_g\psi(Z;\alpha_0,g_0)[\hat{g}-g_0]\mid X=x\Rs  = 0\,, \\
        &\E\Ls \nabla_{\alpha}\psi(Z;\alpha_0,g_0)[\hat{\alpha}-\alpha_0]\mid X=x\Rs=0\,.
        \end{aligned}
    \end{equation}
\end{proposition}

\begin{proposition}[Double Robustness] \label{prop: dr}
    For any fixed estimators $\hat{g}$ for the regression function and $\hat{\alpha}$ for the conditional Riesz representer, 
    \begin{equation}
        \begin{aligned}
            &\E[\psi(Z;\hat\alpha,\hat g)\mid X=x]-\theta_0(x) =\E[( \alpha_0(x,W)-\hat\alpha(x,W))(\hat g(x,W)-g_0(x,W))\mid X=x].
        \end{aligned}
    \end{equation}
\end{proposition}

In Proposition \ref{prop: ortho}, the local orthogonality implies that the conditional expectation of the debiased moment functional is insensitive to local estimation errors in the nuisance functions $\hat g$ and $\hat{\alpha}$ around their true values.
The double robustness in Proposition \ref{prop: dr} ensures the bias of the plug-in estimators of $\psi$ converges to zero at the rate determined by the product of the convergence rates of $\hat\alpha$ and $\hat g$. This property is crucial since it indicates that a fast convergence rate for the estimation of $\theta_0$ can be achieved without requiring highly accurate estimates of either $g_0$ or $\alpha_0$ individually. The double robustness property accommodates complex nuisance regression functions with high-dimensional parameterizations and is particularly desirable in scenarios where accurate estimation of the conditional Riesz representer is achievable.

This debiasing approach is automatic in the sense that, it does not require prior knowledge of the functional form of the debiasing term. In particular, the conditional Riesz representer $\alpha_0$ can be estimated directly as the minimizer of the conditional loss function
\begin{equation}  \label{equ:crr_est}
    \begin{aligned}
    \alpha_0(x,W) =~& \argmin_\alpha \E\Ls (\alpha(x,W)-\alpha_0(x,W))^2 \mid X=x\Rs \\
    =~& \argmin_\alpha \E\Ls \alpha(x,W)^2-2\alpha_0(x,W))\alpha(x,W)+\alpha_0(x,W)^2 \mid X=x\Rs \\
    =~& \argmin_\alpha \E\Ls \alpha(x,W)^2-2m(Z;\alpha) \mid X=x\Rs.
    \end{aligned}
\end{equation}
In Section~\ref{subsubsec: nusaince estimation}, we show that, under appropriate conditions, the estimation of the cRR $\alpha_0$ is naturally aligned with the estimation of the regression function $g_0$ based on the Hilbert projection theorem. Specifically, when $g_0$ follows a partially linear form, $\alpha_0$ can, without loss of generality, be restricted to the same functional form. We leverage this insight to develop an algorithm that automatically estimates the conditional Riesz representer and, in turn, constructs the debiasing term for the original moment functional.

\subsection{Doubly Robust Estimation} \label{subsec: dr estimation}

Now, we introduce our doubly robust procedure for estimating $\theta_0(x)$. We first outline the procedure for estimating $\theta_0(x)$ at a single target point $x$, and then demonstrate how to efficiently estimate $\theta_0(x)$ at a large number of query points.

To estimate $\theta_0(x)$, we begin by splitting the $2n$ observations $Z_i = (Y_i,X_i,W_i)$ into two halves. The second half of the data is used to construct the regression estimator $\hat g$ and the cRR estimator $\hat\alpha$, while the first half is used to compute similarity weights and construct a plug-in estimator of $\theta_0(x)=\E[\psi(Z;\alpha_0,g_0)|X=x]$. Specifically, the procedure consists of the following two steps:
\begin{enumerate}
    \item First step: Construct the regression estimator $\hat g$ and the cRR estimator $\hat\alpha$ using the second half of the data $Z_{n+1:2n}\coloneqq\{Z_{n+1},\ldots, Z_{2n}\}$ with some guarantee on the conditional root mean squared error (to be made precise in Section  \ref{sec: theory}).
    Evaluate $\hat{g}(X_i,W_i)$ and $\hat{\alpha}(X_i,W_i)$ for each $i\in\{1,\ldots,n\}$.
    \item Second step: Compute similarity weights (kernels) $K(x,X_i,Z_{1:n},\xi)$ based on the first half of the data $Z_{1:n}\coloneqq\{Z_1,\ldots, Z_n\}$. 
    Estimate $\theta_0(x)$ via the plug-in weighted average of debiasing moment functional $\psi$ with $\hat{g}(X_i,W_i)$ and $\hat{\alpha}(X_i,W_i)$:
    \begin{align} \label{eq: theta hat}
        \hat{\theta}(x) = \sum_{i=1}^n K(x,X_i,Z_{1:n},\xi)\Ls m(Z_i;\hat{g}(X_i,W_i)) + \hat{\alpha}(X_i,W_i)(Y_i-\hat{g}(X_i,W_i))\Rs.
    \end{align}
\end{enumerate}

In the second step, the weight $K(x,X_i,Z_{1:n},\xi)$ is estimated based on $Z_{1:n}$ and measures the importance of the sample point $X_i$ for estimating $\theta_0(x)$.
In this paper, we consider the sub-sampled kernel that is obtained by averaging $B$ sub-samples of the base kernel $\tilde K$:
\begin{align}
    K(x,X_i,Z_{1:n},\xi) = \frac{1}{B}\sum_{b=1}^B \tilde K(x,X_i,Z_{S_b},\xi_b),
\end{align}
where $\{Z_{S_b}\}_{b=1}^{B}$ are random size-$s$ sub-samples of $\{Z_1,\ldots,Z_n\}$, $\xi_b$ represents the internal randomness of the base kernel, and $\xi$ captures the randomness from both sub-sampling and $\{\xi_b\}_{b=1}^{B}$. The base kernel $\tilde K$ satisfies $\tilde K(x,X_i,Z_{S_b},\xi_b)=0$ for any $i\notin S_b$ and $\sum_{i\in S_b} \tilde K(x,X_i,Z_{S_b},\xi_b) = 1$ for all $b$.

In Section \ref{subsec: drrf}, we propose a data-efficient tree-based algorithm to construct the base kernel $\tilde K$, which directly targets maximizing the heterogeneity of $\theta_0(x)$. Additionally, we demonstrate in Section \ref{subsec: drrf} the use of the forest lasso algorithm for estimating the nuisance functions in the first step of our two-step procedure.

\paragraph{Estimation at multiple points.} In many practical applications, we are interested in estimating $\theta_0(x)$ at a large number of points $x=x_1,\ldots,x_m$. Observe that our estimator in \eqref{eq: theta hat} depends on $x$ only through the kernels \{$K(x,X_i,Z_{1:n},\xi)\}_{i=1}^n$, which measure the similarities of $x$ and the $X_i$'s. Therefore, we can pre-compute and store $\{ m(Z_i;\hat{g}(X_i,W_i)) + \hat{\alpha}(X_i,W_i)(Y_i-\hat{g}(X_i,W_i))\}_{i=1}^n$ based on the training data. To estimate $\theta_0(x_1),\ldots \theta_0(x_m)$, it is then sufficient to evaluate the kernel $K(x,X_i,Z_{1:n},\xi)$ at $x=x_1,\ldots,x_m$ for each $i\in \{1,\ldots,n\}$. Unlike the Orthogonal Random Forest (ORF) \citep{orf}, which requires the estimation of the nuisance functions at each new target point $x$, our procedure allows for the pre-computation and storage of the nuisance function estimates at the sample covariates $X_i$. As a result, when querying multiple points $x_1, \ldots, x_m$, our method avoids re-estimating the nuisance functions, significantly accelerating the estimation process. 
This innovation, however, makes the derivation of the asymptotic normality results for our estimator more challenging, as it requires establishing a uniform bound on nuisance errors. This is accomplished in our proof of Theorem \ref{thm: asymptotic normality}.

\subsection{Doubly Robust Random Forest} \label{subsec: drrf}
In this section, we introduce the Doubly Robust Random Forest (DRRF) for estimating $\theta_0(x)$, which provides an efficient implementation of the doubly robust estimation procedure in Section \ref{subsec: dr estimation}. Specifically, we outline the estimation procedures for the nuisance functions $g_0$ and $\alpha_0$ in Section \ref{subsubsec: nusaince estimation} and detail the construction of the forest kernel in Section \ref{subsubsec: forest kernel}.

\subsubsection{Estimation of \texorpdfstring{$g_0$}{g0} and \texorpdfstring{$\alpha_0$}{alpha0}} \label{subsubsec: nusaince estimation}

In practice, our framework allows for the use of a variety of methods for estimating the nuisance functions $g_0$ and $\alpha_0$. However, to provide a concrete example and establish the foundation for the theoretical results of our estimator, we adopt a specific form of the nuisance regression $g_0$ following \cite{orf}.
Specifically, we assume that $g_0(x,w) = h(w,\nu_g(x))$ for some known function $h$ and unknown function $\nu_g:\mathcal{X}\rightarrow \mathbb{R}^{d_{\nu}}$. We consider the function $h$ such that $\mathcal{H}_x \coloneqq \{g\in \mathcal{G}_x: g(x,w)=h(w,\nu(x)) \ \text{for some function} \ \nu:\mathcal{X}\rightarrow \mathbb{R}^{d_{\nu}} \}$ is a closed linear subspace of $\mathcal{G}_x$. For example, this holds when $h(w, \nu(x)) = r(w)^\top \nu(x)$ for some function $r$, under mild conditions.\footnote{See Appendix \ref{app: riesz} for a proof.} We allow the dimension $d_{\nu}$ to grow with $n$.

Consequently, the conditional Riesz representer $\alpha_0(x, w)$ can be projected onto the space $\mathcal{H}_x$, yielding $\alpha_{\mathcal{H}_x}(x,w)$ such that for all \( g \in \mathcal{H}_x \),
\[
\E[m(Z; g(X, W)) \mid X = x] = \E[g(X, W) \alpha_{\mathcal{H}_x}(X, W) \mid X = x].
\]

Therefore, without loss of generality, we can consider the conditional Riesz representer $\alpha_0(w,x)$ as taking the form $h(w,\nu_\alpha(x))$, with the same function $h$ and an unknown function $\nu_\alpha:\mathcal{X}\rightarrow \mathbb{R}^{d_{\nu}}$.

Suppose that the unknown functions $\nu_g$ and $\nu_\alpha$ can be identified as the minimizers of the following local losses with known loss functions $l_g$ and $l_\alpha$:
\begin{equation}
    \begin{aligned}
        &\nu_g(x) = \argmin_{\nu\in V}\E[l_g(Z;\nu)\mid X=x],\\
        &\nu_\alpha(x) = \argmin_{\nu\in V}\E[l_\alpha(Z;\nu)\mid X=x],
    \end{aligned}
\end{equation}
where $V$ is a bounded space. 

\begin{example}[Locally Linear Parameterization] Consider the case where the function $h$ is linear in $w$ locally around each $x$, taking the form $h(w, \nu(x)) = w^\top \nu(x)$. Since $g_0$ is the conditional expectation of the outcome $Y$, we can define the loss function $l_g$ for $\nu_g$ as the squared loss $l_g(Z;\nu)=(Y-W^\top \nu)^2$. Moreover, the conditional Riesz representer satisfies $\E[f_j(X,W)\alpha_0(X,W)- m(Z;f_j(X,W))\mid X=x]=0$ for all $f_j(x,w)\coloneqq w_j$ in $\mathcal{H}_x$. Let $\tilde{m}(Z;W)=(m(Z;f_1),\ldots, m(Z;f_{d_\nu}))^\top$. Under the covariance condition $\E[WW^\top \mid X=x] \succ 0$, we can take the loss function $l_\alpha$ for $\nu_\alpha$ to be $l_\alpha(Z;\nu)=(W^\top \nu)^2 - 2\tilde{m}(Z;W)^\top \nu$.
\end{example}

Next, we demonstrate the use of the Forest Lasso algorithm proposed by \cite{orf} to estimate $g_0$ and $\alpha_0$ in the first step of the doubly robust estimation procedure. Specifically, we run the tree learner defined in Section \ref{subsubsec: forest kernel} on $\{Z_{n+1},\ldots, Z_{2n}\}$ based on $B$ subsamples of size $s$ to obtain kernels $K(\cdot,X_i,Z_{n+1:2n},\xi)$ for $i = n+1, \ldots, 2n$. Then, we solve the lasso problems
\begin{equation}  \label{equ: lasso}
    \begin{aligned}
    \hat{\nu}_g(\cdot) &= \argmin_{\nu\in V}\sum_{i=n+1}^{2n} K(\cdot,X_i,Z_{n+1:2n},\xi) l_g(Z_i;\nu) + \lambda_g \|\nu\|_1,\\
    \hat{\nu}_\alpha(\cdot) &= \argmin_{\nu\in V}\sum_{i=n+1}^{2n} K(\cdot,X_i,Z_{n+1:2n},\xi) l_\alpha(Z_i;\nu) + \lambda_{\alpha} \|\nu\|_1
    \end{aligned}
\end{equation}
 at $X_1,\ldots,X_n$. In contrast to \cite{orf}, which applies Forest Lasso at the target point $x$, our approach runs the algorithm at the training data points $X_1,\ldots,X_n$. This modification eliminates the need to rerun the algorithm for each new target point once we solve the $n$ Forest Lasso programs associated with $X_1,\ldots,X_n$. 

 As to be shown in Section \ref{sec: theory}, with appropriately chosen regularization parameters $\lambda_g$ and $\lambda_{\alpha}$, the estimators $\hat{\nu}_g$ and $\hat{\nu}_\alpha$ are consistent under the assumptions that the high-dimensional functions $\nu_g$ and $\nu_\alpha$ are sparse and supported on a finite set of dimensions, along with other regularity conditions. These results are formalized in Lemma \ref{pro: nu guarantee} in Appendix, which will then be used to establish the asymptotic normality of the target parameter $\theta_0(x)$ in Theorem \ref{thm: asymptotic normality for DRRF}.

\subsubsection{Forest kernel} \label{subsubsec: forest kernel}

We proceed to describe the forest algorithm for computing the similarity weights $K(x,X_i,Z_{1:n},\xi)$ in \eqref{equ: lasso} and in the second step of our estimation. As stated in Section \ref{subsec: dr estimation}, the kernel $K(x,X_i,Z_{1:n},\xi)$ is the average over $B$ base kernels $\tilde K(x,X_i,Z_{S_b},\xi_b)$ constructed from random size-$s$ sub-samples $S_b$ for $b=1,\ldots,B$. Our algorithm combines the debiased moment functional with the tree-based learner from ORF \citep{orf} to compute the base kernel $\tilde K$. 

For each sub-sample $S_b$, we randomly partition it into two equal-sized subsets $S_b^1$ and $S_b^2$. We use $S_b^1$ to grow a tree learner that splits the feature space $\mathcal{X}$ into small subspaces, and use $S_b^2$ to evaluate the similarity weights.

\paragraph{Splitting algorithm.}
For $S_b^1$, we start from a root node containing the entire space $\mathcal{X}$. At each node $P$, we implement the following three-stage procedure: 
\begin{enumerate}
\item Estimate $\nu_g$ by solving an equally weighted Lasso problem
\[
\hat{\nu}_g = \argmin_{v\in V}\frac{1}{|P\cap S_b^1|}\sum_{i\in P\cap S_b^1}l_g(Z_i;v) + \lambda_g \|v\|_1.
\] Estimate $\nu_\alpha$ using the same method.\footnote{In practice, one can add a small ridge penalty to the minimization problems to improve stability, especially when the sample size \( |P \cap S_b^1| \) is small. In our implementation, we apply this approach when estimating $\nu_{\alpha}$.}
\item For each $i\in P\cap S_b^1$, compute the nuisance estimates $\hat{g}(X_i,W_i)=h(W_i, \hat{\nu}_g)$ and $\hat{\alpha}(X_i,W_i)=h(W_i, \hat{\nu}_\alpha)$.
\item For each potential split, calculate $\hat{\theta}$ at the resulting two children nodes $C_1$ and $C_2$ as:
\begin{align*}
    \hat{\theta}_{C_k} = \frac{1}{|C_k\cap S_b^1|}\sum_{i \in C_k\cap S_b^1} \Ls m(Z_i;\hat{g}(X_i,W_i))+\hat{\alpha}(X_i,W_i)(Y_i-\hat{g}(X_i,W_i))\Rs, \quad k=1,2.
\end{align*}
Denote $\psi_i =m(Z_i;\hat{g}(X_i,W_i))+\hat{\alpha}(X_i,W_i)(Y_i-\hat{g}(X_i,W_i))$. 
Find the split that minimizes the sum of squares error (SSE) $\sum_{i\in C_1 \cap S_b^1} (\hat{\theta}_{C_1}-\psi_i)^2+\sum_{i\in C_2 \cap S_b^1} (\hat{\theta}_{C_2}-\psi_i)^2.$\footnote{This is equivalent to finding a split for the standard regression tree with response $\psi_i$ for each sample point. Note that minimizing the SSE is also equivalent to maximizing the heterogeneity score $\frac{1}{|C_1\cap S_b^1|}\big(\sum_{i\in C_1 \cap S_b^1} (\psi_i-\hat{\theta}_P)\big)^2+\frac{1}{|C_2\cap S_b^1|}\big(\sum_{i\in C_2 \cap S_b^1} (\psi_i-\hat{\theta}_P)\big)^2$, where $\hat{\theta}_P$ is the estimated $\hat\theta$ at the parent node.}
\end{enumerate}
Throughout the tree-growing process, we maintain the tree properties specified in Assumption \ref{assump: forest}.

\paragraph{Weight estimation.}
For each tree based on $S_b^1,$ let $L_b(x)$ denote the leaf that contains the target value $x$. We compute the base kernels $\tilde K(x,X_i,Z_{S_b},\xi_b)$ for each point $i$ as
\begin{align*}
    \tilde K(x,X_i,Z_{S_b},\xi_b) = \frac{\mathbbm{1}((i\in S_b^2)\cap (X_i\in L_b(x)))}{|S_b^2 \cap L_b(x)|}.
\end{align*}
Note that we only assign positive weights $\tilde K(x,X_i,Z_{S_b},\xi_b)$ to the sample points in $S_b^2$.

\section{Theoretical Results} \label{sec: theory}
In this section, we establish the consistency and asymptotic normality of the general doubly robust estimator introduced in Section~\ref{subsec: dr estimation}, as well as the DRRF estimator described in Section~\ref{subsec: drrf}, which provides an effective implementation of the general estimator. These results are derived under appropriate regularity conditions.

\subsection{Consistency}

We begin by establishing the consistency of the doubly robust estimator $\hat{\theta}(x)$, constructed using sub-sampled kernels as described in Section~\ref{subsec: dr estimation}, for estimating the target parameter $\theta_0(x)$. To ensure consistency, we require that the base kernel $\tilde K$ satisfies the following properties:

\begin{assumption}[Honesty]  \label{assump: honesty}
For any set $S_b$ and sample $i$, the base kernel $\tilde{K}(x,X_i,Z_{S_b},\xi_b)$ satisfies
\[
\tilde{K}(x,X_i,Z_{S_b},\xi_b) \perp\!\!\!\perp m(Z_i;\hat{g}(X_i,W_i))+\hat{\alpha}(X_i,W_i)(Y_i-\hat{g}(X_i,W_i)) \mid X_i, \{Z_j\}_{j\in S_b /\{i\}}\,.
\]
\end{assumption}

\begin{assumption}[Kernel Shrinkage] \label{assump: kernel shrinkage}
For any size-$s$ subset $S_b$ of the data, 
\[
\E\Ls \max_{i} \Big\{\|X_i - x\|: \tilde{K}(x,X_i,Z_{S_b},\xi_b)>0\Big\}\Rs \leq \epsilon(s)\,.
\]
\end{assumption}
Assumption \ref{assump: honesty} ensures the conditional independence of the base kernel and the estimated debiased moment functional. Assumption \ref{assump: kernel shrinkage} states that, in expectation, positive weights are assigned only to sample points whose features $X_i$ are sufficiently close to the target $x$. It is straightforward to verify that Assumption \ref{assump: honesty} holds for the DRRF algorithm defined in Section \ref{subsec: drrf}. Moreover, following \cite{wager2018estimation}, the DRRFs that satisfy the forest regularity conditions specified in Assumption \ref{assump: forest} achieve a kernel shrinkage rate of $\epsilon(s) = O(s^{-1/(2a d)})$,\footnote{The parameter $a=\log(\rho^{-1})/(\pi\log((1-\rho)^{-1}))\,,$ where $\rho$ represents the smallest fraction of leaves on each side of the split, and $\pi/d$ is the smallest probability that the split occurs along each feature, both as defined in Assumption \ref{assump: forest}.} which decays polynomially in $s$. 

Finally, we impose additional conditions on the debiased moment functional, as formalized in Assumption~\ref{assump: moment}. These conditions introduce standard boundedness and smoothness assumptions, similar to those assumed in, for example, \cite{orf}.

\begin{assumption} \label{assump: moment} 
\begin{enumerate}
    \item{(Boundedness).} \label{subassump: boundedness} There exists a bound $\psi_{\max}$ such that $|\psi(Z_i;\alpha,g)|\leq \psi_{\max}$ for any observation $Z_i$ and functions $g,\alpha \in \mathcal{G}_x$.
    \item{(Smoothness).} \label{subassump: Lip of moments} 
    The moment $\E\Ls  \psi(Z;\alpha,g) \mid X=x \Rs$ is $L$-Lipschitz in $x$, uniformly over all functions $g,\alpha \in \mathcal{G}_x$.
\end{enumerate}
\end{assumption}

Now we formally state the consistency result for the general doubly robust estimator $\hat{\theta}(x)$ proposed in Section \ref{subsec: dr estimation}.
\begin{theorem}[Consistency of the General Doubly Robust Estimator] \label{thm: consistency}
Define $\mathcal{E}(\hat{g},\hat{\alpha})= \|\hat{g}-g_0\|_x\cdot \|\hat{\alpha}-\alpha_0\|_x.$ Suppose Assumptions \ref{assump: honesty}, \ref{assump: kernel shrinkage}, and \ref{assump: moment} hold. Then, with probability $1-2\delta$,
\begin{equation} \label{equ:thm1_1}
\La\hat{\theta}(x)-\theta_0(x)\Ra \leq  \sqrt{\frac{2\log(2/\delta)}{B}}\psi_{\max} + \sqrt{\frac{2s\log(2/\delta)}{n}}\psi_{\max}+L\epsilon(s)+\mathcal{E}(\hat{g},\hat{\alpha}).
\end{equation}
Furthermore, suppose $s\rightarrow \infty, s=o(n), B\geq n/s$, $\epsilon(s)\rightarrow 0$, and $\mathcal{E}(\hat{g},\hat{\alpha})\rightarrow 0$ as $n\rightarrow \infty$. Then, as $n\rightarrow\infty$, the doubly robust estimator $\hat\theta(x)$ consistently estimates $\theta_0(x)$ for any fixed $x$:
\begin{align*}
    \hat{\theta}(x)-\theta_0(x)=o_p(1).
\end{align*}
\end{theorem}

In Theorem \ref{thm: consistency}, the right-hand side of inequality~\eqref{equ:thm1_1} decomposes into four terms, each reflecting a distinct source of error. The first term captures the variation introduced by aggregating over the $B$ sub-samples used to construct the base kernels. The second term represents the sampling error associated with an appropriate complete U-statistic based on the observed data. The third term reflects the error coming from kernel estimation, where $\epsilon(s)$ denotes the kernel shrinkage rate specified in Assumption~\ref{assump: kernel shrinkage}. Finally, the fourth term accounts for the estimation errors in the regression function and the conditional Riesz representer. Note that the convergence rate of $\hat{\theta}(x)$ depends on the nuisance functions solely through the product $\mathcal{E}(\hat{g},\hat{\alpha}) = \Lv \hat{g}-g_0\Rv_x \cdot \Lv \hat{\alpha}-\alpha_0\Rv_x.$ Therefore, the proposed estimation procedure is indeed doubly robust with respect to nuisance estimation, and it achieves the oracle rate as long as $\mathcal{E}(\hat{g},\hat{\alpha})$ converges to 0 at a slow rate.

We now establish the consistency of the doubly robust estimator $\hat{\theta}(x)$ estimated by our DRRF algorithm introduced in Section \ref{subsec: drrf}. Following the discussion in Section \ref{subsubsec: nusaince estimation}, we focus on nuisance estimators of the form $\hat{g}(x,w)=h(w,\hat{\nu}_g(x))$ and $\hat{\alpha}(x,w)=h(w,\hat{\nu}_\alpha(x))$. 
To provide theoretical guarantees for the local regularized nuisance estimators in \eqref{equ: lasso}, we impose the following Assumption \ref{assump: nuisance} on the true parameters $\nu_g$ and $\nu_\alpha$.

\begin{assumption}[Nuisance Regularity] \label{assump: nuisance}
The parameters and data distribution satisfy:
\begin{enumerate}
    \item For $\forall x\in \mathcal{X}$, $\nu_g(x)$ is $k$-sparse with support $S_g(x)$, and $\nu_\alpha(x)$ is $k$-sparse with support $S_\alpha(x)$.
    \item Both $\nu_g(x)$ and $\nu_\alpha(x)$ are $O(1)$-Lipschitz in $x$. The function $h(\cdot,\cdot)$ is $O(1)$-Lipschitz in its second argument. Furthermore, the functions $\nabla_\nu L_g(x;\nu)=\E\Ls \nabla_\nu l_g(Z;\nu)\mid X=x\Rs$ and $\nabla_\nu L_\alpha(x;\nu)=\E\Ls \nabla_\nu l_\alpha(Z;\nu)\mid X=x\Rs$ are $O(1)$-Lipschitz in $x$ for any $\nu$.
    \item The data distribution satisfies the conditional restricted eigenvalue condition: for all $\nu$ and for all $z$, for some matrices $\mathcal{H}_{g}(z)$ and $\mathcal{H}_{\alpha}(z)$ that depend only on the data: $\nabla_{\nu\nu}l_g(z;\nu) \succeq \mathcal{H}_g(z) \succeq 0$ and $\nabla_{\nu\nu}l_\alpha(z;\nu) \succeq \mathcal{H}_{\alpha}(z) \succeq 0$. Let $C(S(x);3) = \{\nu\in \mathbb{R}^d: \|\nu_{S(x)^c}\|_1\leq 3\|\nu_{S(x)}\|_1\}$. For all $x$, $\nu_g\in C(S_g(x);3)$, and $\nu_{\alpha}\in C(S_g(x);3)$, $\nu_g^\top \E\Ls \mathcal{H}_g(Z)\mid X=x\Rs \nu_g \geq \gamma \|\nu_g\|_2^2$ and $\nu_{\alpha}^\top \E\Ls \mathcal{H}_{\alpha}(Z)\mid X=x\Rs \nu_{\alpha} \geq \gamma \|\nu_{\alpha}\|_2^2$.
\end{enumerate}
\end{assumption}

At a high level, Assumption 4 ensures that the nuisance parameters $\nu_g$ and $\nu_\alpha$ are sparse, and can be reliably estimated from the data under standard identifiability conditions. This assumption is important for proving that the nuisance estimators $\hat{g}(x,w)=h(w,\hat{\nu}_g(x))$ and $\hat{\alpha}(x,w)=h(w,\hat{\nu}_\alpha(x))$, defined via the penalized regression problems in \eqref{equ: lasso} in Section \ref{subsubsec: nusaince estimation}, converge to the true functions $g_0$ and $\alpha_0$, respectively.
Under this assumption, we establish the consistency result for DRRF estimator.

\begin{corollary}[Consistency of DRRF] \label{coro:consistency_DRRF}
Suppose that Assumptions \ref{assump: moment} and \ref{assump: nuisance} hold, and that DRRF satisfies the forest regularity conditions defined in Appendix \ref{assump: forest}. Let the subsample size be chosen as $s=\Theta(n^{\beta})$, where $\Lp 1+\frac{1}{a d}\Rp^{-1}<\beta<1$, $B\ge n/s$, and the parameters $$\lambda_g,\lambda_\alpha = \Theta\Lp s^{-1/(2a d)}+\sqrt{\frac{s\ln(n^{\frac{1}{2}(1-\beta)} d_{\nu})}{n}}\Rp.$$ Then, as $n\rightarrow\infty$, the DRRF algorithm consistently estimates $\theta_0(x)$ for any fixed $x$:
\begin{align*}
    \hat{\theta}(x)-\theta_0(x) &= O_p\Lp n^{\frac{1}{2}(\beta-1)}\Rp.
\end{align*}
\end{corollary}

Corollary \ref{coro:consistency_DRRF} shows that the DRRF algorithm achieves the same consistency rate as estimating the conditional functional with oracle nuisance functions. The convergence rate $\beta$ depends on the forest regularity parameter $a = \log(\rho^{-1}) / (\pi \log((1 - \rho)^{-1}))$ (see Assumption \ref{assump: forest}) and the dimension $d$ of the target feature space $\mathcal{X}$. Specifically, the rate improves with more balanced tree splits (i.e., larger $\rho$ and $\pi$) and lower dimension $d$, given suitable choices of the subsample size $s$ and the number of base kernels $B$.

\subsection{Asymptotic Normality} \label{subsec:asymp norm}
In this section, we establish the asymptotic normality results for the general doubly robust estimator and DRRF. We again focus on nuisance estimators of the form $\hat{g}(x,w)=h(w,\hat{\nu}_g(x))$ and $\hat{\alpha}(x,w)=h(w,\hat{\nu}_\alpha(x))$. Additionally, we assume that the dimensionality $d_{\nu}$ of $\nu_g$ and $\nu_\alpha$ grows at most polynomially in $n$.
We require an additional assumption on the debiased moment functional:
\begin{assumption} \label{assump: msc} 
{(Local Mean-Squared Continuity).}  
Define $V(x,g,\alpha) \coloneqq \E[(\psi(Z;\alpha,g) -\psi(Z;\alpha_0,g_0) )^2\mid X=x]$.
For any $x\in \mathcal{X}$ and $g,\alpha \in \mathcal{G}_x$, the moment $V(x,g,\alpha)$ satisfies
\begin{align*}
    |V(x;g,\alpha)|\le~& L \cdot \E\Ls(g(X,W)-g_0(X,W))^2+(\alpha(X,W)-\alpha_0(X,W))^2\mid X=x\Rs\,.
\end{align*}
\end{assumption}

Assumption \ref{assump: msc} imposes a local Lipschitz-type condition on the moment functional $\psi(Z;\alpha,g)$ with respect to the nuisance functions $g$ and $\alpha$. Specifically, it ensures that the conditional squared error in the moment functional arising from estimation of $g$ and $\alpha$ can be bounded by the corresponding estimation errors of these functions. This assumption is imposed for the asymptotic results in Theorems~\ref{thm: asymptotic normality} and \ref{thm: asymptotic normality for DRRF}, as it guarantees that the sampling error associated with estimated nuisances can be asymptotically dominated by the term involving the oracle nuisances $\alpha_0$ and $g_0$, whose asymptotic variance involves the term
$$\eta(s)\coloneqq Var\Lp \E \Ls \sum_{i=1}^s\tilde{K}(x,X_i,Z_S,\xi_S)\psi(Z_i;\alpha_0,g_0)\bigg|Z_1 \Rs\Rp\,,$$ 
where the set $S$ is fixed to include the first $s$ samples. 

We now state the asymptotic normality result for the general doubly robust estimator proposed in Section \ref{subsec: dr estimation}.

\begin{theorem}[Asymptotic Normality of the General Doubly Robust Estimator] 
\label{thm: asymptotic normality}

Suppose Assumptions \ref{assump: honesty}, \ref{assump: kernel shrinkage}, \ref{assump: moment}, and \ref{assump: msc} hold. Let the subsample size $s$ grow such that $s\rightarrow \infty$, $s=o(n)$, $B\ge (n/s)^2$, $\epsilon(s)\rightarrow 0$, $n\eta(s) \rightarrow \infty$, and $n\epsilon^2(s)/(s^2\eta(s))\rightarrow 0$ as $n\rightarrow \infty$. 
Moreover, for any fixed $x$, with probability at least $1-\delta$,
$$\Lp\|\hat{\nu}_g(x) - \nu_g(x)\|_2^2 + \|\hat{\nu}_\alpha(x) - \nu_\alpha(x)\|_2^2\Rp^{1/2}\leq r_n(\delta).$$
Suppose there exists a sequence $\{\delta_n\}_n$ such that 
$\sqrt{n/(s^2\eta(s))}\cdot \delta_n \rightarrow 0$ and $\sqrt{n/(s^2\eta(s))}\cdot r_n^2(\delta_n/n)\rightarrow 0.$
Then, as $n\rightarrow\infty$, the asymptotic distribution of the doubly robust estimator $\hat{\theta}(x)$ is given by 
$$\sqrt{\frac{n}{s^2\eta(s)}}\Lp \hat{\theta}(x)-\theta_0(x)\Rp \xrightarrow{d}N(0,1)\,.$$ 
\end{theorem}

In Theorem \ref{thm: asymptotic normality}, the conditions $n\eta(s) \rightarrow \infty$ and $n\epsilon^2(s)/(s^2\eta(s))\rightarrow 0$ require $\eta(s)$ to decay slower than $1/n$ and $\epsilon^2(s)/(s^2\eta(s))$ to decay faster than $1/n$. This can be achieved under the ``double-sample'' scheme in the DRRF algorithm under an appropriate choice of the subsample size $s$.
Moreover, the nuisance rate condition $\sqrt{n/(s^2\eta(s))}\cdot r_n^2(\delta_n/n)\rightarrow 0$ in Theorem \ref{thm: asymptotic normality} is usually mild and only requires the nuisance parameters $\nu_g$ and $\nu_\alpha$ to converge at a slow rate. Following Theorem \ref{thm: asymptotic normality}, we immediately obtain the following asymptotic normality result for DRRF.

\begin{theorem}[Asymptotic Normality of DRRF] 
\label{thm: asymptotic normality for DRRF}
Suppose Assumptions \ref{assump: moment}, \ref{assump: msc}, \ref{assump: nuisance}, and tree regularity conditions in \ref{assump: forest} hold. Set the subsample size $s=\Theta(n^{\beta})$, where $\Lp 1+\frac{1}{a d}\Rp^{-1}<\beta<1$, $B\ge (n/s)^2$, and the lasso parameters $\lambda_g,\lambda_\alpha = \Theta\Lp s^{-1/(2a d)}+\sqrt{\frac{s\ln(n^c d_{\nu})}{n}}\Rp$ for some $c>1/2$. Furthermore, suppose there exists a strictly positive $\underline\sigma$ such that $min_x Var(\psi(Z;\alpha_0,g_0)\mid X=x)\ge\underline{\sigma}^2>0$. Then, as $n\rightarrow\infty$, the asymptotic distribution of $\hat{\theta}(x)$ estimated by the DRRF algorithm is given by $$\sqrt{\frac{n}{s^2\eta(s)}}\Lp \hat{\theta}(x)-\theta_0(x)\Rp \xrightarrow{d}N(0,1)\,.$$ 
\end{theorem}

In light of Theorem \ref{thm: asymptotic normality for DRRF}, we can use the Bootstrap of Little Bags \citep{sexton2009standard, grf} approach to construct asymptotically valid confidence intervals for $\theta_0(x)$. In Section \ref{sec:simulations}, we show that this construction provides good empirical coverage.

\section{Simulations}
\label{sec:simulations}
In this section, we conduct comprehensive simulations to evaluate the performance of our DRRF estimator in heterogeneous treatment effect (HTE) estimation, comparing it against other tree-based non-parametric methods.

We generate data according to the following process:
\[
Y_i = \theta_0(X_i)\cdot D_i+ \tilde X_i^\top \nu_0 + \epsilon_i,
\]
where the target covariates $X_i \overset{i.i.d.}{\sim} U[0,1]$, the noise terms $\epsilon_i \overset{i.i.d.}{\sim} U[-1,1]$, and the coefficient vector $\nu_0$ is $k$-sparse, with non-zero coefficients drawn i.i.d. from $U[-1,1]$. Following \cite{orf}, we define the true heterogeneous treatment effect $\theta_0(x)$ to be a piecewise linear function: \[
\theta_0(x) = 
\begin{cases} 
x + 2 & \text{if } x \leq 0.3, \\ 
6x + 0.5 & \text{if } 0.3 < x \leq 0.6, \\ 
-3x + 5.9 & \text{if } x > 0.6.
\end{cases}
\] 
Our goal is to estimate $\theta_0(x)$ at various query points $x$.
We consider different data-generating processes for the treatment variable $D_i$ and the control covariates $\tilde{X}_i$. In the main text, we focus on the case where the treatment variable $D_i$ is binary, while Appendix \ref{appendix:simu} presents results for continuous treatment settings. The total number of observations is fixed at $2n = 5000$, with $p = 100$ additional covariates. The support size $k$ varies across $\{5, 10, 15\}$. A detailed description of the data-generating processes is provided in Table \ref{tab:res_binary}.

We compare the performance of two variants of DRRF, along with two variants each of ORF \citep{orf} and GRF \citep{grf}:
\begin{itemize}
    \item {\bf DRRF variants:} (1) DRRF: Our automatic DRRF algorithm, described in Section \ref{subsec: drrf}, with parameters set based on theoretical guidance: subsample power $\beta=0.88$, number of trees $B=100$, split balance parameter $\rho=0.3$, minimum leaf size $r=5$, and $L_1$ regularization parameters $\lambda_g,\lambda_\alpha=\sqrt{s\ln(n(p+1))/n}/5$. (2) DRRF-CV: A variant of DRRF where the $L_1$ regularization parameter $\lambda_g$ for the Forest Lasso estimator $\hat{\nu}_g$ in \eqref{equ: lasso} is chosen by 5-fold cross-validation.
    \item {\bf ORF variants:} (1) ORF: The ORF algorithm from \cite{orf} with fixed Lasso penalty parameters for tree splitting and Forest Lasso, as specified in \cite{orf}.  (2) ORF-CV: A cross-validated variant of ORF where $L_1$ regularization parameters in Forest Lasso are chosen using 5-fold cross-validated Lasso.
    \item {\bf GRF variants:} (1) ForestDR-Linear: A combination of GRF with a doubly robust moment, using cross-validated linear models regularized by $L_1$ or $L_2$ norms to estimate the conditional outcomes $\E[Y\mid D=d, X, \tilde X]$ and propensity scores $\PP(D=d\mid X,\tilde X)$. This is implemented by the ForestDRLearner in the EconML package \citep{econml}, and it is applicable only to the binary treatment setting. (2) ForestDR-RF: Similar to ForestDR-Linear, this variant combines a doubly robust moment with GRF. However, it employs cross-validated Random Forests instead of the linear models to estimate the conditional outcomes and propensity scores.
\end{itemize}

There are two key differences among the three variants. First, both DRRF and ORF variants accommodate high-dimensional nuisance functions and perform local nuisance estimation. 
In contrast, GRF variants apply tree learner-generated similarity weights to moment functions constructed with global nuisance estimators. Second, while both DRRF and GRF variants estimate the target based on a doubly robust moment, ORF variants use an orthogonal moment without the doubly robust property. On the other hand, constructing the doubly robust moment in GRF variants requires propensity score estimation, which restricts their applicability to continuous treatment settings. In contrast, our DRRF variants automatically construct the doubly robust moment without requiring propensity score estimation, making them applicable to both discrete and continuous treatment settings.

\begin{table}[htbp]
\centering
\begin{tabular}{c|ccc|ccc}
\toprule
\multicolumn{1}{c}{} & \multicolumn{3}{c|}{\bf RMSE} &\multicolumn{3}{c}{\bf Time (seconds)}\\
\midrule
\multicolumn{1}{c}{} &\multicolumn{6}{c}{\textbf{Setup 1}}\\
\cmidrule(r){2-7}
\multicolumn{1}{c}{} &$k=5$ &$k=10$ &$k=15$  &$k=5$ &$k=10$ &$k=15$ \\
\midrule
 \bf DRRF &0.126 &0.134 &0.149 &0.17 &0.16 &0.16\\
 \bf DRRF-CV &\bf 0.108 &\bf 0.110 &\bf 0.108 &0.17 &0.17 &0.17\\
\midrule
 ORF &0.196 &0.201 &0.201&0.54 &0.54 &0.54\\
 ORF-CV &\bf 0.108 &0.119 &0.130&21.0 &20.5 &19.6\\
 \midrule
 ForestDR-Linear &0.279 &0.288 &0.281&0.10 &0.09 &0.09\\
 ForestDR-RF &0.212 &0.316 &0.404&0.08 &0.08 &0.09\\
\midrule
\multicolumn{1}{c}{} &\multicolumn{6}{c}{\textbf{Setup 2}}\\
\midrule
 \bf DRRF &0.136 &0.151 &0.161&0.16 &0.16 &0.17\\
 \bf DRRF-CV &\bf0.113 &\bf0.119 &\bf0.115&0.17 &0.17 &0.17\\
\midrule
 ORF &0.157 &0.157 &0.156 &0.51 &0.51 &0.57\\
 ORF-CV &0.160 &0.174 &0.189 &20.8 &20.0 &19.6\\
 \midrule
 ForestDR-Linear &0.345 &0.352 &0.345 &0.09 &0.09 &0.09\\
 ForestDR-RF &0.238 &0.351 &0.453 &0.09 &0.09 &0.09\\
\midrule
\multicolumn{1}{c}{} &\multicolumn{6}{c}{\textbf{Setup 3}}\\
\midrule
 \bf DRRF &0.126 &0.148 &0.156 &0.17 &0.16 &0.17\\
 \bf DRRF-CV &\bf0.110 &\bf0.116 &\bf0.106&0.16 &0.17 &0.17\\ 
\midrule
 ORF &0.129 &0.174 &0.196 &2.6 &2.8 &3.0\\
 ORF-CV &0.184 &0.242 &0.293 &120.9 &118.5 &138.7\\ 
 \midrule
 ForestDR-Linear &0.348 &0.329 &0.338 &0.09 &0.09 &0.09\\
 ForestDR-RF &0.246 &0.379 &0.512&0.09 &0.08 &0.09\\
\bottomrule
\end{tabular}
\caption{RMSE and runtime for binary treatment setups}
\floatfoot{This table reports the root mean squared error (RMSE) and the evaluation runtime for different estimators in binary treatment setups. We generate data from the process $Y_i=\theta_0(X_i)\cdot D_i+\tilde X_i^\top \nu_0+\epsilon_i$, where $X_i\overset{i.i.d.}{\sim} U[0,1]$ and $\nu_0$ is a $k$-sparse vector with non-zero coefficients drawn i.i.d. from $U[-1,1]$. The binary treatment variables $D_i$'s are generated according to the propensity scores $\PP(D_i=1\mid X_i,\tilde X_i)$. The propensity scores and the additional covariates $\tilde X_i$ are generated as follows: (1) Setup 1: $\PP(D_i=1\mid X_i,\tilde X_i) = 0.3+0.4\cdot I(X_i>0.5)$, and $\tilde{X}_i \sim \mathcal{N}(0, I_p)$. (2) Setup 2: $\PP(D_i=1\mid X_i,\tilde X_i) = 0.1+0.3\cdot I(\tilde{X}_i^\top \gamma_0 > 0)+0.4\cdot I(X_i>0.5)$, and $\tilde{X}_i \sim \mathcal{N}(0, I_p)$. (3) Setup 3:  $\PP(D_i=1\mid X_i,\tilde X_i) = 0.2+0.6\cdot I(\tilde{X}_i^\top \gamma_0 > 3)$, and $\tilde{X}_i \sim \mathcal{N}(3, I_p)$. In Setup 2, $\gamma_0$ is a $k$-sparse vector with the same support as $\nu_0$, and its non-zero coefficients are drawn i.i.d. from $U[0,1]$. In Setup 3, $\gamma_0$ is a $k$-sparse vector with the same support as $\nu_0$ and non-zero coefficients equal to $1/k$. We set the number of covariates $\tilde X_i$ to be $p=100$ and the sample size $2n=5000$. We perform 100 simulations for each setup. The evaluation runtime is the average runtime over 10 runs for estimating $\hat{\theta}(x_i)$ at 100 test points $x_i$ over a uniform grid in $[0,1]$. The results are measured on a 2019 MacBook Pro with a 2.6 GHz 6-Core Intel Core i7.}
\label{tab:res_binary}
\end{table}

Table \ref{tab:res_binary} compares the estimation accuracy and evaluation runtime of different estimators. Specifically, it reports the root mean squared error (RMSE) of the estimated $\hat{\theta}(x_i)$ and the average evaluation runtime across 100 test points $x_i$ over a uniform grid in $[0,1]$. We report results based on 100 simulations for each setup. Figures \ref{fig:binary_1}-\ref{fig:binary_3} illustrate the performance of each estimator by plotting the average estimated effects, along with the 5th and 95th percentiles of the estimated effects across 100 replications. Additionally, Figure \ref{fig:CI} provides an example of the bootstrap confidence intervals for the DRRF algorithm, demonstrating reliable finite-sample coverage.

In the following, we highlight three key advantages of DRRF variants. First, similar to ORF variants, DRRF variants perform local fitting of nuisance functions, which generally leads to superior performance compared to GRF variants, which rely on global fitting. The global nuisance estimation in GRF variants optimizes overall mean squared error, which does not align with the objective of final-stage local estimation. As reported in Table \ref{tab:res_binary}, both DRRF variants exhibit significantly lower RMSE than GRF variants in the binary treatment setting. In fact, Figures \ref{fig:binary_1}-\ref{fig:binary_3} show that GRF variants suffer from substantially higher variance. This high variance may also be linked to their dependence on inverse propensity score estimates, which can become unstable when the propensity scores are near boundary values.

Second, DRRF variants do not require a well-fitted model for the treatment variable, making them more robust than ORF variants. As shown in Table \ref{tab:res_binary}, DRRF variants outperform ORF variants in RMSE in most cases. Note that in the three setups considered, the relationship between treatment $D_i$ and covariates $\tilde{X}_i$ is non-linear. In such cases, the Lasso-based residualization in ORF variants struggles to adequately residualize $D$ from $\tilde{X}$, degrading the performance of ORF variants. This issue is more pronounced in the continuous treatment setting, as shown in Table \ref{tab:res_cts} in Appendix \ref{appendix:simu}.  In contrast, our DRRF algorithms automatically construct a debiasing term without relying on a well-specified model between $D$ and $\tilde{X}$.

Finally, DRRF variants are significantly more computationally efficient than ORF variants, especially ORF-CV, when estimating $\theta_0(x)$ at multiple query points. 
As shown in Table \ref{tab:res_binary}, the evaluation runtimes of both DRRF and DRRF-CV are substantially shorter than those of ORF and ORF-CV. In particular, DRRF-CV’s average evaluation runtime is over 100 times faster than ORF-CV. This demonstrates the efficiency of the DRRF algorithm and makes it a more practical choice for real-world applications requiring fast estimations at a large number of query points.

Collectively, these results demonstrate the advantages of the proposed DRRF algorithms in terms of estimation accuracy, robustness, and computational efficiency, underscoring their potential for practical applications.

\begin{figure}[htbp]
    \centering
    \begin{subfigure}{\textwidth}
        \centering
        \includegraphics[width=0.78\linewidth,height=5.2cm]{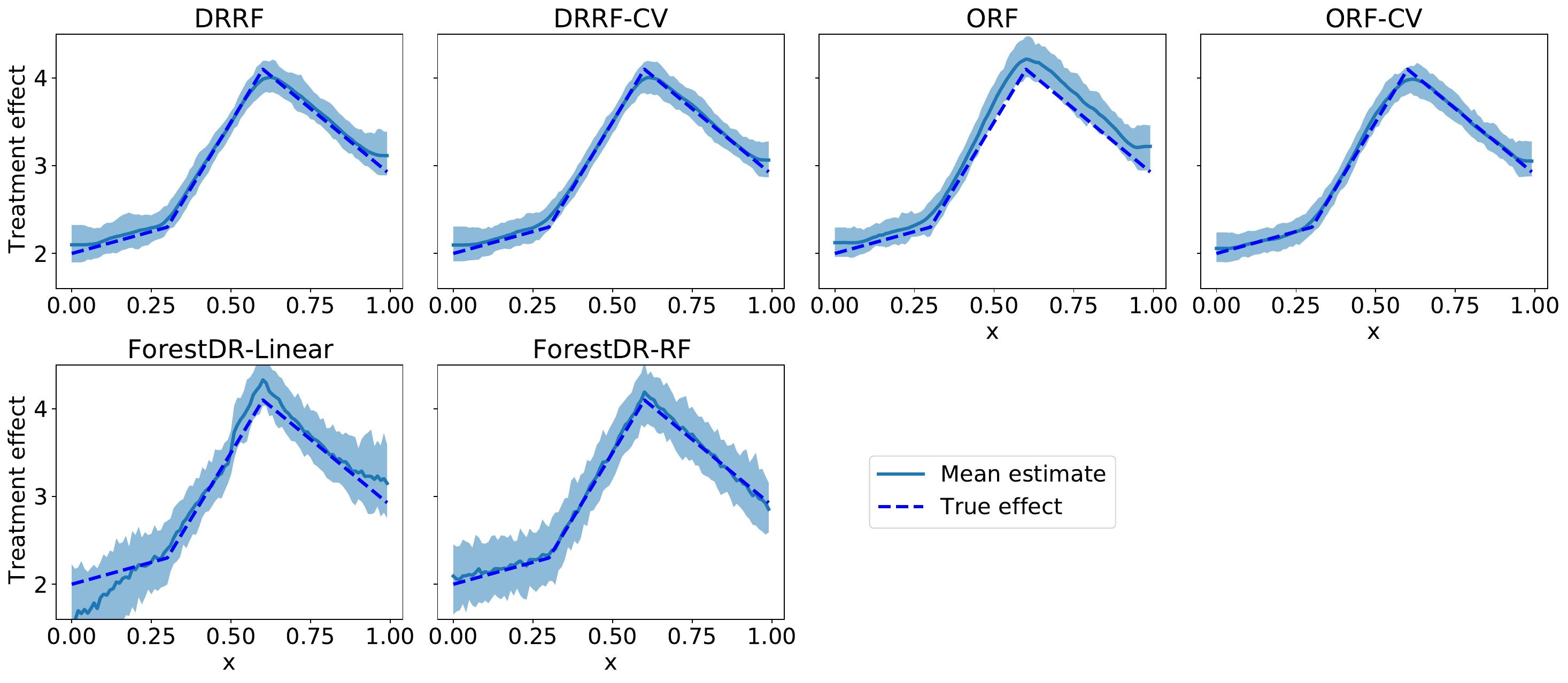} 
        \caption{$k=5$}
    \end{subfigure}
    \begin{subfigure}{\textwidth}
        \centering
        \includegraphics[width=0.78\linewidth,height=5.2cm]{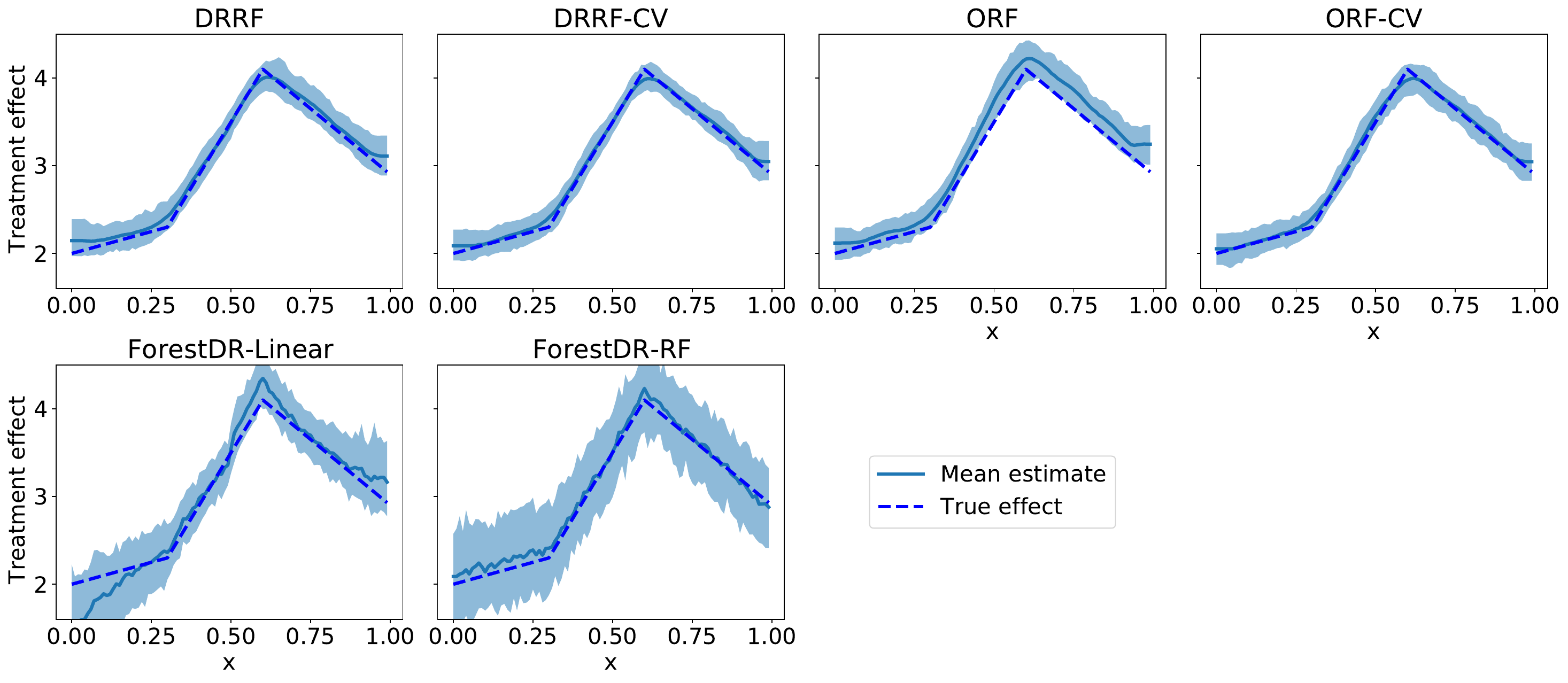} 
        \caption{$k=10$}
    \end{subfigure}
    \begin{subfigure}{\textwidth}
        \centering
        \includegraphics[width=0.78\linewidth,height=5.2cm]{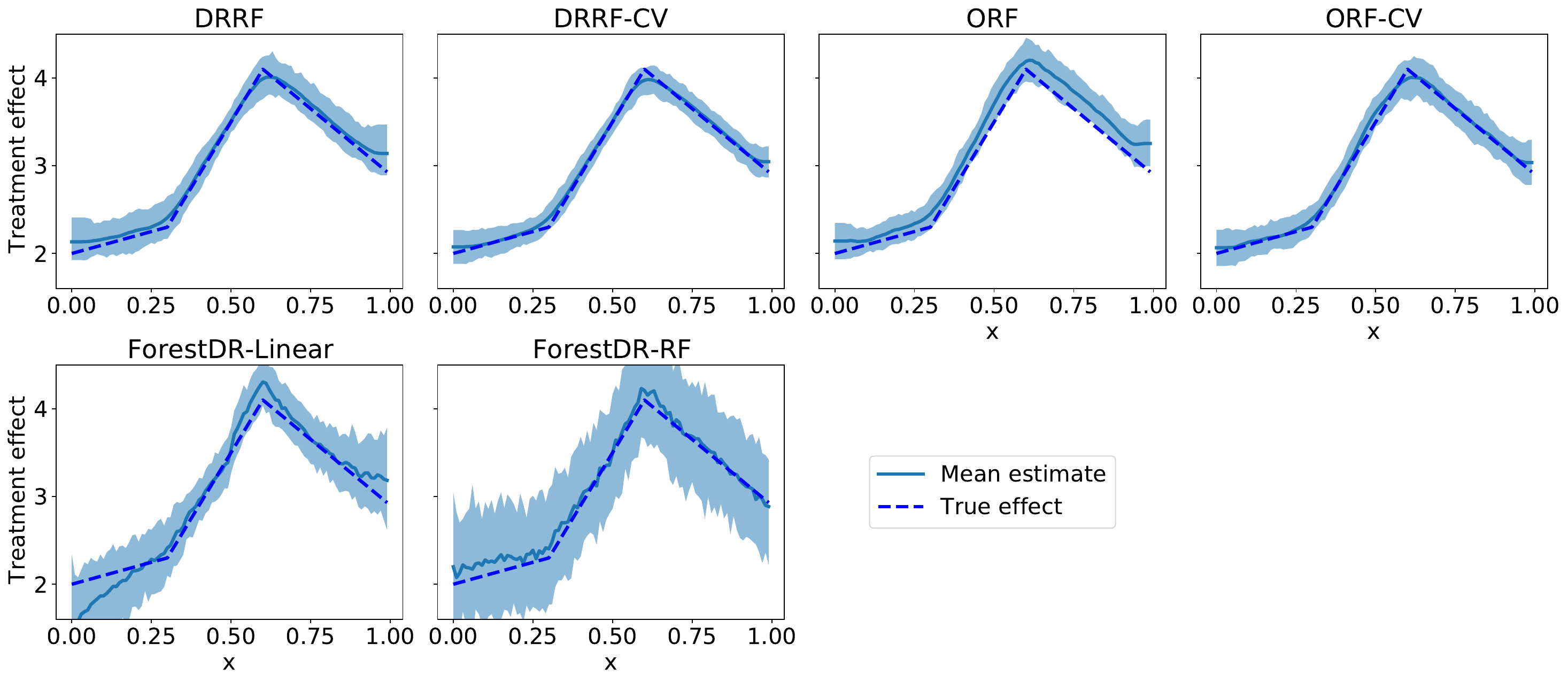} 
        \caption{$k=15$}
    \end{subfigure}
    \caption{Treatment effect estimations for Setup 1 in binary treatment setting}
    \floatfoot{This figure shows the treatment effect estimations for Setup 1 in the binary treatment setting. The plot shows the average estimated effects, with shaded regions representing the 5th to 95th percentile intervals of the estimated effects across 100 replications. The data is generated from the process $Y_i=\theta_0(X_i)\cdot D_i+\tilde X_i^\top \nu_0+\epsilon_i$, where $X_i\overset{i.i.d.}{\sim} U[0,1]$, $\tilde{X}_i \sim \mathcal{N}(0, I_p)$, and $\gamma_0$ is a $k$-sparse vector with non-zero coefficients drawn i.i.d. from $U[-1,1]$. The binary treatment variable $D_i$ is randomly generated based on the propensity score $\PP(D_i=1\mid X_i,\tilde X_i) = 0.3+0.4\cdot I(X_i>0.5)$. We set the number of observations to $2n=5000$ and $p=100$.}
    \label{fig:binary_1}
\end{figure}

\begin{figure}[htbp]
    \centering
    \begin{subfigure}{\textwidth}
        \centering
        \includegraphics[width=0.78\linewidth,height=5.2cm]{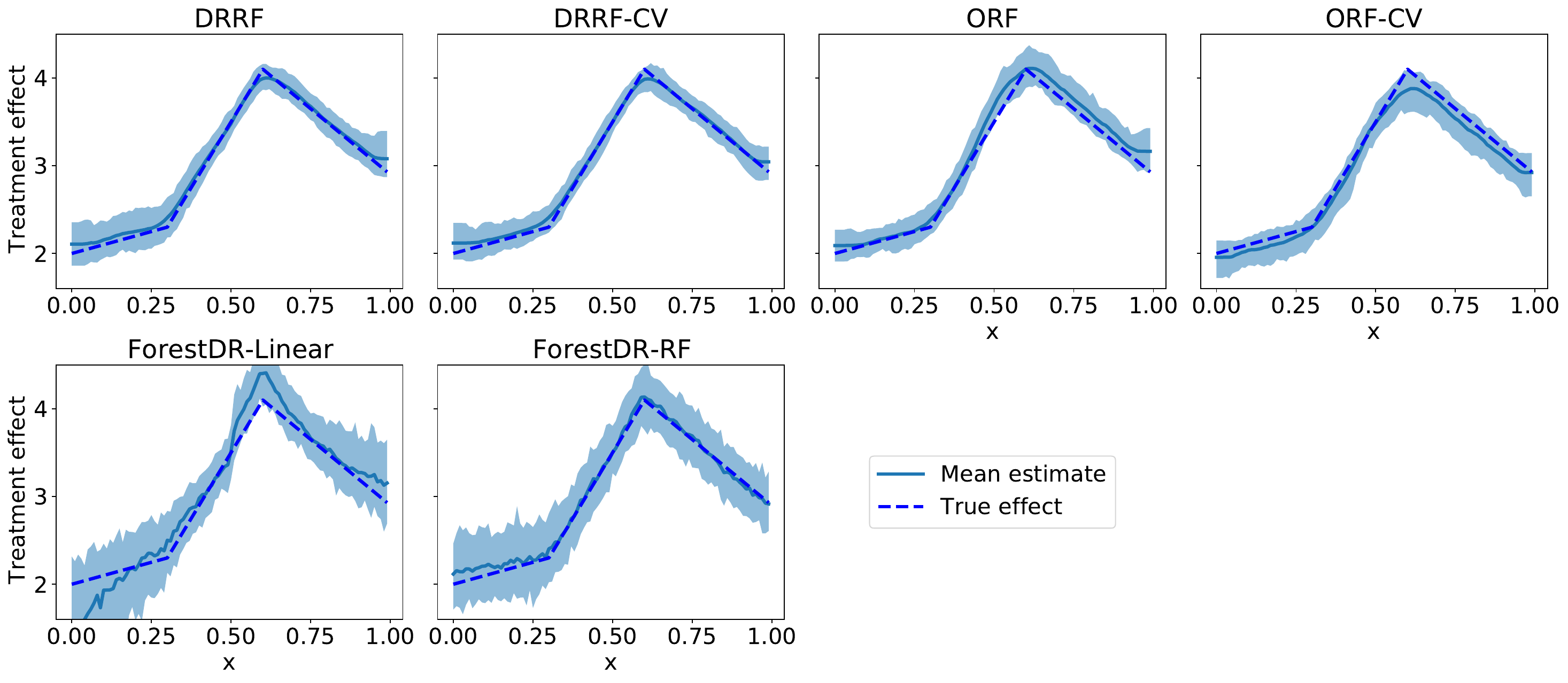} 
        \caption{$k=5$}
    \end{subfigure}
    \begin{subfigure}{\textwidth}
        \centering
        \includegraphics[width=0.78\linewidth,height=5.2cm]{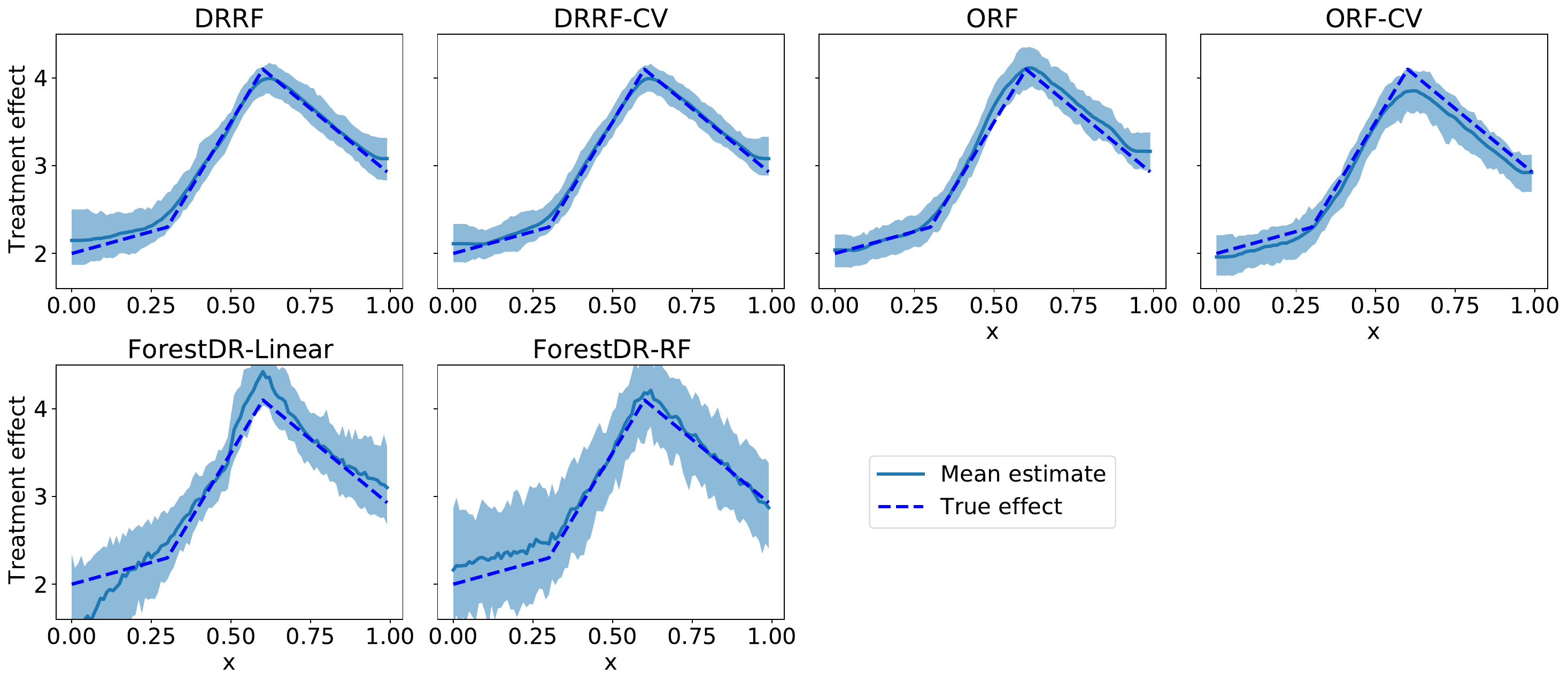} 
        \caption{$k=10$}
    \end{subfigure}
    \begin{subfigure}{\textwidth}
        \centering
        \includegraphics[width=0.78\linewidth,height=5.2cm]{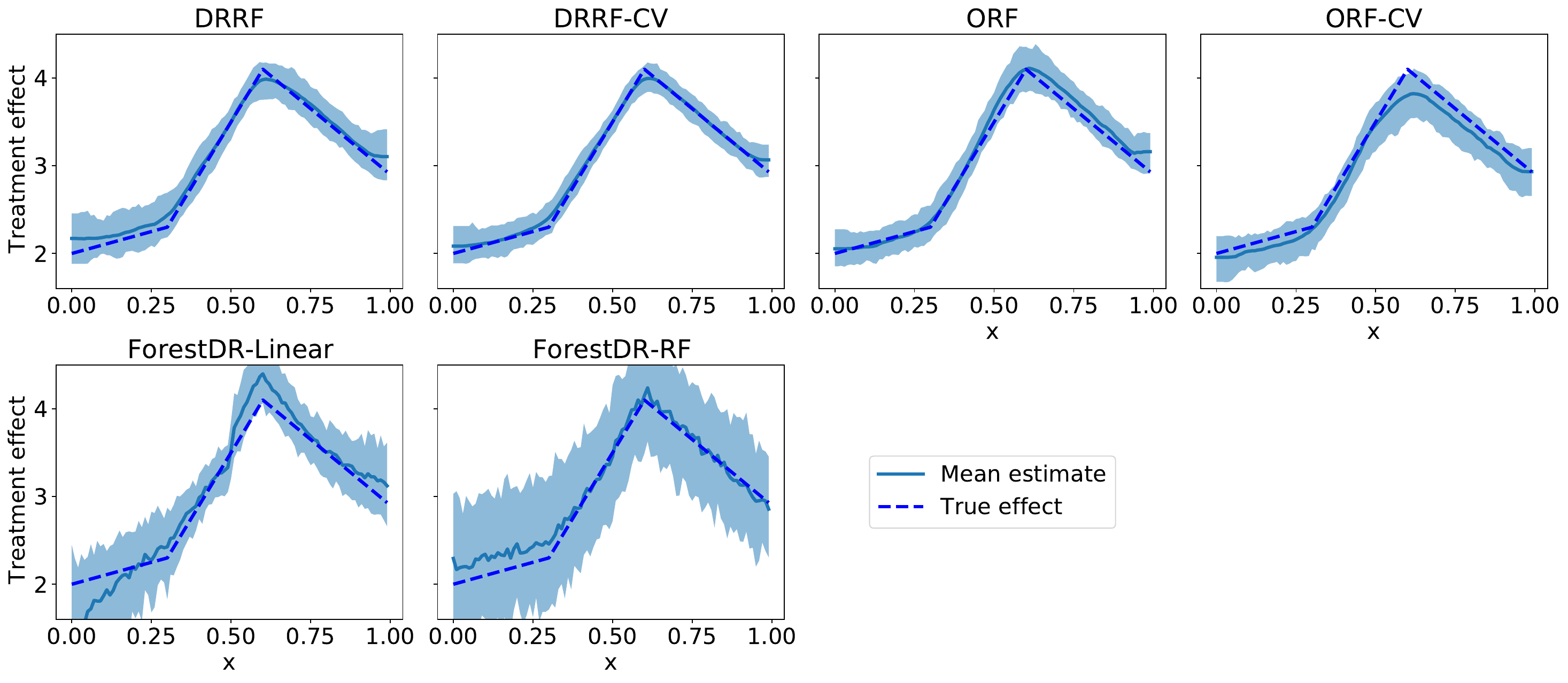} 
        \caption{$k=15$}
    \end{subfigure}
    \caption{Treatment effect estimations for Setup 2 in binary treatment setting}
    \floatfoot{This figure shows the treatment effect estimations for Setup 2 in the binary treatment setting. The plot shows the average estimated effects, with shaded regions representing the 5th to 95th percentile intervals of the estimated effects across 100 replications. The data is generated from the process $Y_i=\theta_0(X_i)\cdot D_i+\tilde X_i^\top \nu_0+\epsilon_i$, where $X_i\overset{i.i.d.}{\sim} U[0,1]$, $\tilde{X}_i \sim \mathcal{N}(0, I_p)$, and $\gamma_0$ is a $k$-sparse vector with non-zero coefficients drawn i.i.d. from $U[-1,1]$. The binary treatment variable $D_i$ is randomly generated based on the propensity score $\PP(D_i=1\mid X_i,\tilde X_i) = 0.1+0.3\cdot I(\tilde X_i^\top \gamma_0 > 0)+0.4\cdot I(X_i>0.5)$, where $\gamma_0$ is a $k$-sparse vector with the same support as $\nu_0$, and its non-zero coefficients are drawn i.i.d. from $U[0,1]$. We set the number of observations to $2n=5000$ and $p=100$.}
    \label{fig:binary_2}
\end{figure}

\begin{figure}[htbp]
    \centering
        \begin{subfigure}{\textwidth}
        \centering
        \includegraphics[width=0.78\linewidth,height=5.2cm]{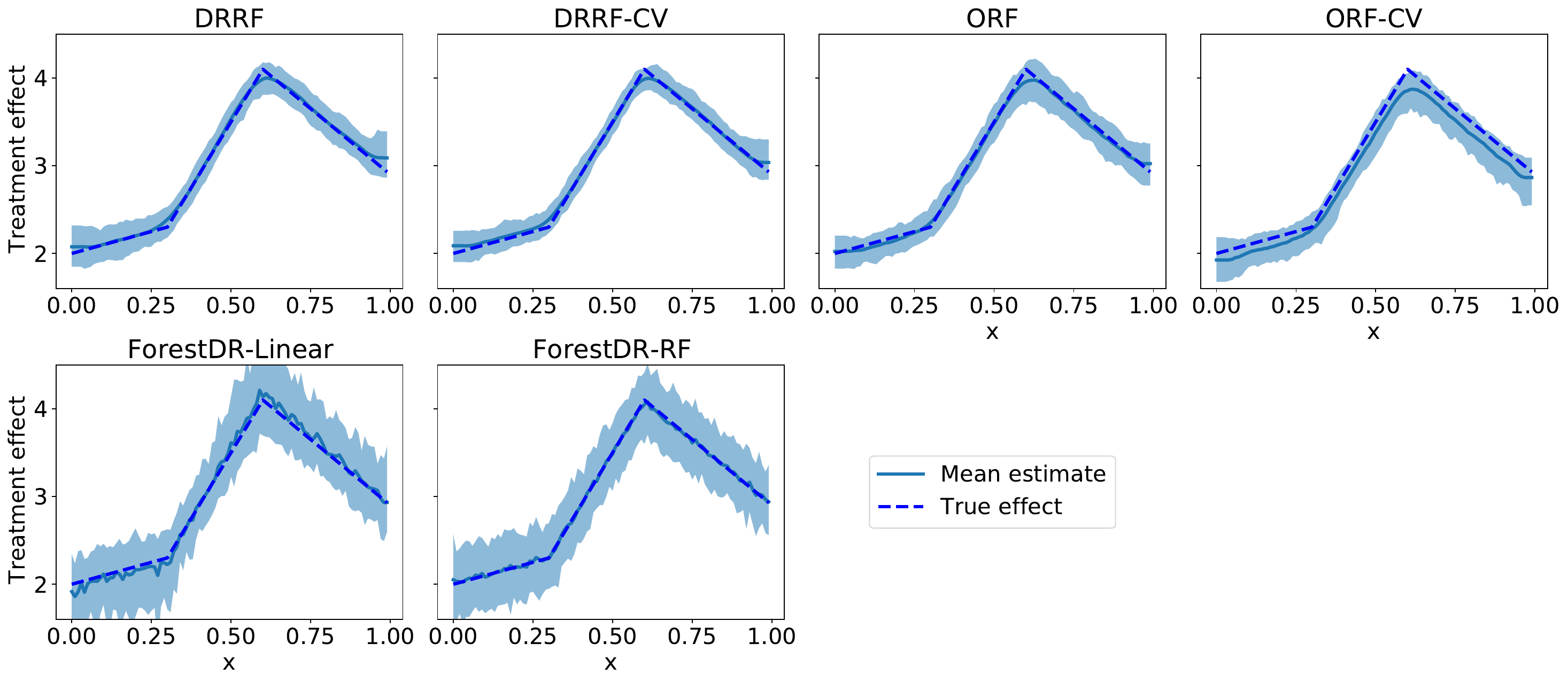} 
        \caption{$k=5$}
    \end{subfigure}
    \begin{subfigure}{\textwidth}
        \centering
        \includegraphics[width=0.78\linewidth,height=5.2cm]{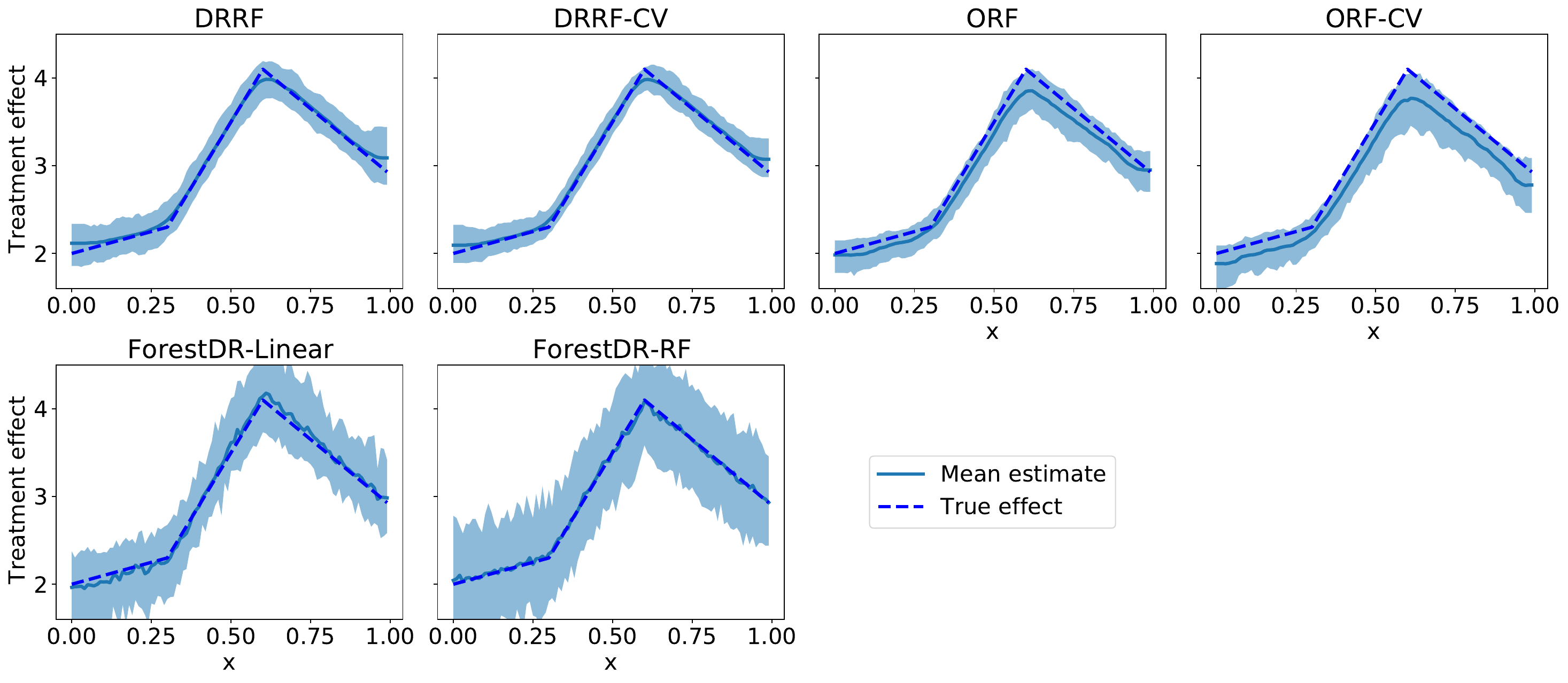} 
        \caption{$k=10$}
    \end{subfigure}
    \begin{subfigure}{\textwidth}
        \centering
        \includegraphics[width=0.78\linewidth,height=5.2cm]{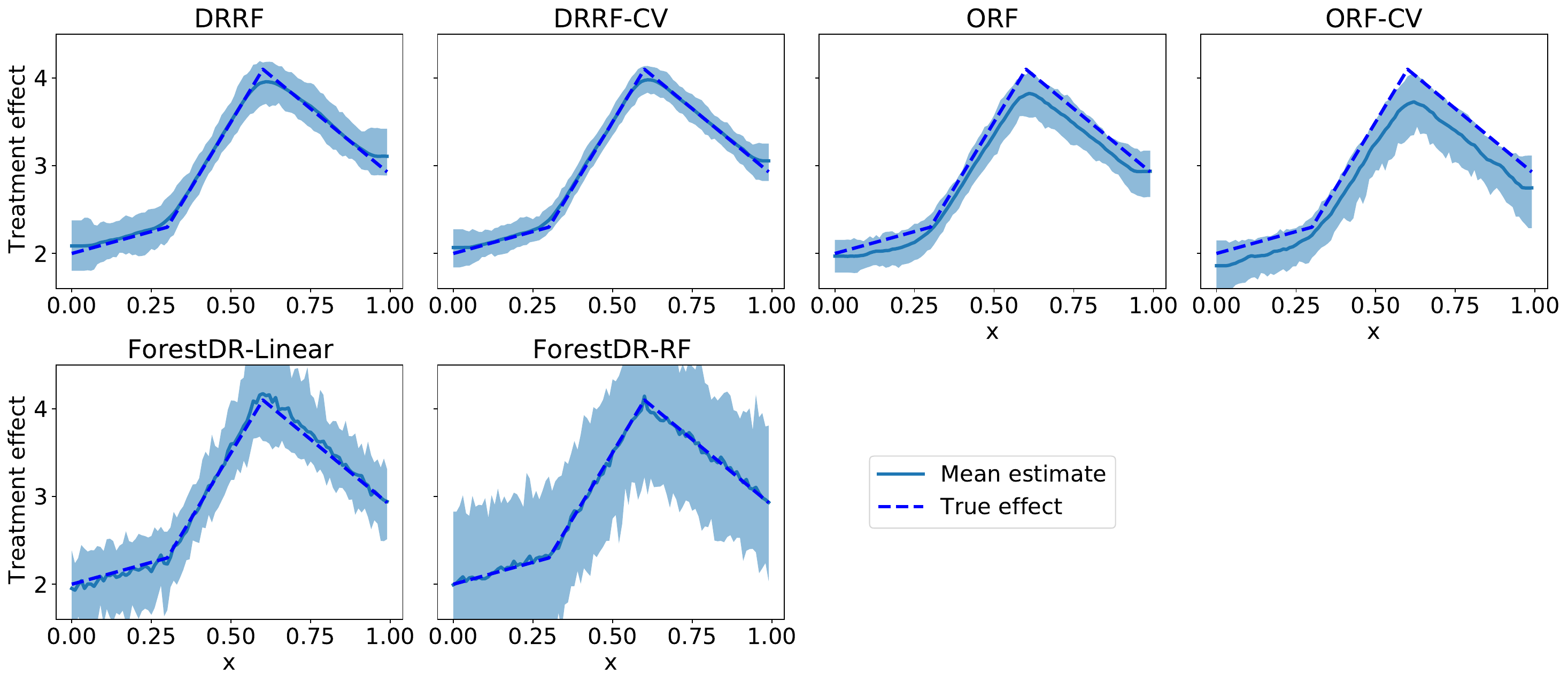} 
        \caption{$k=15$}
    \end{subfigure}
    \caption{Treatment effect estimations for Setup 3 in binary treatment setting}
    \floatfoot{This figure shows the treatment effect estimations for Setup 3 in the binary treatment setting. The plot shows the average estimated effects, with shaded regions representing the 5th to 95th percentile intervals of the estimated effects across 100 replications. The data is generated from the process $Y_i=\theta_0(X_i)\cdot D_i+\tilde X_i^\top \nu_0+\epsilon_i$, where $X_i\overset{i.i.d.}{\sim} U[0,1]$, $\tilde{X}_i \sim \mathcal{N}(3, I_p)$, and $\gamma_0$ is a $k$-sparse vector with non-zero coefficients drawn i.i.d. from $U[-1,1]$. The binary treatment variable $D_i$ is randomly generated based on the propensity score $\PP(D_i=1\mid X_i,\tilde X_i) = 0.2+0.6\cdot I(\tilde X_i^\top \gamma_0 > 0)$, where $\gamma_0$ is a $k$-sparse vector with the same support as $\nu_0$ and non-zero coefficients equal to $1/k$. We set the number of observations to $2n=5000$ and $p=100$.}
    \label{fig:binary_3}
\end{figure}

\begin{figure}[t!]
    \centering
    \includegraphics[width=0.4\linewidth]{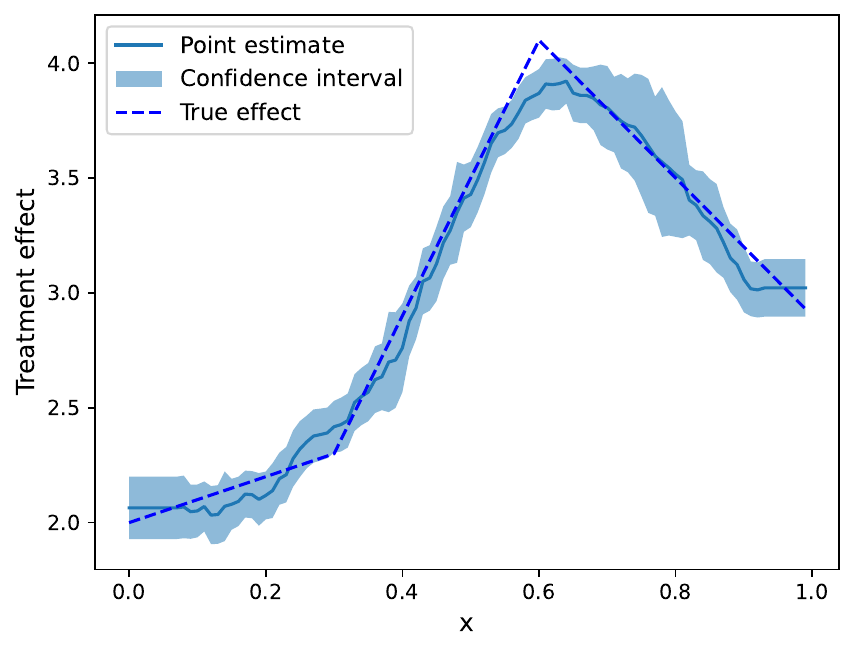}
    \caption{Example of bootstrap confidence intervals for DRRF}
    \label{fig:CI}
    \floatfoot{This figure presents the 90\% confidence intervals constructed using the bootstrap of little bags algorithm for Setup 1 in the binary treatment setting. We set the support size to $k=5$ and the number of trees to $B=100$. Approximately 90\% of the sampled test points are contained in the intervals.}
\end{figure}

\section{Conclusion}
In this paper, we proposed an efficient and automatic doubly robust method for estimating conditional moment functionals, introducing the Doubly Robust Random Forest (DRRF) as a practical algorithm. Our approach extended the existing automatic debiased machine learning (Auto-DML) framework, based on Riesz representers and originally developed for average functionals, to the more general setting of conditional functionals. In particular, we introduced the notion of conditional Riesz representer (cRR) and showed that, under mild conditions, cRR can be projected onto the same function space as the regression function. In particular, when the regression function is locally linear, it suffices to consider a locally linear conditional Riesz representer, which can be efficiently estimated via a weighted Lasso. Our method also offers substantial computational gains at evaluation time compared to existing approaches such as \cite{orf} by precomputing and storing moment estimates from the training data. We established the consistency and asymptotic normality of our procedure under appropriate regularity conditions. Overall, our method provides a nonparametric, doubly robust, and computationally scalable solution for estimating general conditional moment functionals in the presence of high-dimensional nuisance functions.

\newpage
\bibliographystyle{plainnat} 
\bibliography{ref} 

\newpage
\setcounter{section}{0}
\renewcommand\thesection{A.\arabic{section}} 
\setcounter{assumption}{0}
\setcounter{lemma}{0}
\setcounter{proposition}{0}
\setcounter{table}{0}
\setcounter{figure}{0}
\renewcommand{\theassumption}{A.\arabic{assumption}}
\renewcommand{\thelemma}{A.\arabic{lemma}}
\renewcommand{\theproposition}{A.\arabic{proposition}}
\renewcommand{\thefigure}{A.\arabic{figure}}
\renewcommand{\thetable}{A.\arabic{table}}
\section*{Appendix}

\section{Technical Assumptions}
\begin{assumption}[Forest Regularity]  \label{assump: forest}
The tree satisfies
\begin{enumerate} 
    \item $\rho-$balanced: Each split leaves at least a fraction $\rho$ of the observations in $S_b^2$ on each side of the split for some parameter of $\rho \le 0.2$.
    \item Minimum leaf size r: There are between $r$ and $2r-1$ observations from $S_b^2$ in each leaf of the tree.
    \item $\pi-$random-split: At every step, marginalizing over the internal randomness of the learner, the probability that the next split occurs along the $j$-th feature is at least $\pi/d$ for some $0<\pi\le1$, for all $j=1,...,d$.
\end{enumerate}
\end{assumption}

\section{Sufficient Conditions for the Form of Conditional Riesz Representer} \label{app: riesz}
In this section, we give sufficient conditions with which we can assume without loss of generality that $\alpha_0(w,x)=h(w,\nu_\alpha(x))$ if $g_0(x,w)=h(w,\nu_g(x))$.

Recall that $\mathcal{G}_x = \{g:\E[g(X,W)^2\mid X=x]<\infty\}.$ 
Since $\mathcal{G}_x$ is the $L_2$ space under the law $\mathcal{L}([X,W]\mid X=x)$, it is a Hilbert space with inner product defined by $$\langle g_1,g_2\rangle_x = \E[g_1(X,W)g_2(X,W)|X=x]. $$

Suppose that $\mathcal{H}_x$ is a closed linear subspace of $\mathcal{G}_x$. Let $\alpha_0 \in \mathcal{G}_x$ be the conditional Riesz representer. Then by Hilbert space projection theorem, we can write $\alpha_0 = \alpha_{\mathcal{H}_x} + \alpha_{\mathcal{H}_x}^{\perp}$, where $\alpha_{\mathcal{H}_x}\in \mathcal{H}_x$ and $\alpha_{\mathcal{H}_x}^{\perp} \in \mathcal{H}_x^{\perp}$. As a result, for any $g \in \mathcal{H}_x$,
\begin{equation}  \label{equ: Riesz}
    \begin{aligned}
        \E[m(Z;g(X,W))\mid X=x]&=\E[g(X,W)\alpha_0(X,W)\mid X=x]\\ &=\E[g(X,W)\alpha_{\mathcal{H}_x}(X,W)\mid X=x].
    \end{aligned}
\end{equation}
Since throughout our theoretical analysis, we only work with those $g\in \mathcal{H}_x$, we can assume without loss of generality that $\alpha_0(x,w)$ also takes form of $h(w,\nu_{\alpha}(x))$ for some $\nu_\alpha(x)$. In fact, it is the minimal norm conditional Riesz representer such that \eqref{equ: Riesz} holds for all $g \in \mathcal{H}_x$.

Now we show that with $h(w,\nu)=w^\top \nu$, $\mathcal{H}_x$ is a closed linear subspace of $\mathcal{G}_x$, under mild conditions.  Since $$h(w,\nu_1(x))+ch(w,\nu_2(x)) = w^\top (\nu_1(x)+c\nu_2(x))\in \mathcal{H}_x$$ for any $c\in \mathbb{R},$ $\mathcal{H}_x$ is a linear subspace of $\mathcal{G}_x.$
Suppose that $h(w,\nu_1(x)), h(w,\nu_2(x)),\ldots$ is a sequence of elements in $\mathcal{H}_x$ that converges to a limit $g^*(x,w) \in \mathcal{G}_x.$ Then the sequence is a Cauchy sequence. If follows that $$(\nu_n(x)-\nu_m(x))^\top\mathbb{E}[WW^\top\mid X=x](\nu_n(x)-\nu_m(x)) \xrightarrow[]{} 0$$ as $m,n\xrightarrow[]{} \infty.$
Assume that $\mathbb{E}[WW^\top\mid X=x]$ is positive definite. Then above implies that $\norm{\nu_n(x)-\nu_m(x)}_2\xrightarrow[]{} 0$ as $m,n\xrightarrow[]{} \infty.$ Therefore $\{\nu_n(x)\}_n$ is a Cauchy sequence that converges to some limit $\nu^*(x)\in \mathbb{R}^{d_\nu}.$ 
Since $$\norm{h(w,\nu_n(x))-h(w,\nu^*(x))}^2_{\mathcal{G}_x}=(\nu_n(x)-\nu^*(x))^\top\mathbb{E}[WW^\top\mid X=x](\nu_n(x)-\nu^*(x)) \xrightarrow[]{} 0$$ as $n \xrightarrow[]{} \infty,$ we get that $h(w,\nu_n(x)) \xrightarrow[]{} h(w,\nu^*(x))$ in the norm induced by the Hilbert space $\mathcal{G}_x.$ Therefore $\mathcal{H}_x$ is a closed linear subspace of $\mathcal{G}_x.$ It is easy to see that the above argument extends straightforwardly to the case where $h(w,\nu)=r(w)^\top \nu$ for some known function $r$.

\section{Proofs}
\subsection{Proof of Proposition \ref{prop: ortho}}
The local gradient of $\psi(Z;\alpha_0, g_0) = m(Z;g_0)+\alpha_0(Y- g_0)$ with respect to $g$ is
\begin{align*}
&\E\Ls \nabla_g\psi(Z;\alpha_0,g_0)[\hat{g}-g_0]\mid X=x\Rs \\
=~& \E\Ls\lim_{h\rightarrow 0} \frac{1}{h}\Lp\psi(Z;\alpha_0,g_0+h(\hat g-g_0))-\psi(Z;\alpha_0,g_0)\Rp \mid X=x\Rs\\
=~& \E\Ls\lim_{h\rightarrow 0} \frac{1}{h}\Lp m(Z;g_0+h(\hat g-g_0))-m(Z;g_0)\Rp \mid X=x\Rs-\E\Ls \alpha_0(\hat g-g_0)\mid X=x\Rs.
\end{align*}
Under regularity conditions,\footnote{For example, it suffies to assume that $m$ is Lipschitz in its argument $g$ and $\hat{g}\in \mathcal{G}_x$.} the limit and expectation are exchangeable. Then by the definition of the cRR, the first part on the RHS can be further simplified as
\begin{align*}
    \E\Ls\lim_{h\rightarrow 0} \frac{1}{h}\Lp m(Z;g_0+h(\hat g-g_0))-m(Z;g_0)\Rp \mid X=x\Rs 
    =~& \lim_{h\rightarrow 0}\frac{1}{h}\E\Ls (g_0+h(\hat g-g_0))\alpha_0-g_0\alpha_0 \mid X=x\Rs\\
    =~& \E\Ls \alpha_0(\hat g-g_0)\mid X=x\Rs,
\end{align*}
which implies that $\E\Ls \nabla_g\psi(Z;\alpha_0,g_0)[\hat{g}-g_0]\mid X=x\Rs=0$. Similarly, the local gradient of $\psi(Z;\alpha_0, g_0)$ with respect to $\alpha$ is
\begin{align*}
\E\Ls \nabla_{\alpha}\psi(Z;\alpha_0,g_0)[\hat{\alpha}-\alpha_0]\mid X=x\Rs
=~& \E\Ls\lim_{h\rightarrow 0} \frac{1}{h}\Lp\psi(Z;\alpha_0+h(\hat \alpha-\alpha_0),g_0)-\psi(Z;\alpha_0,g_0)\Rp \mid X=x\Rs\\
=~& \E\Ls (\hat{\alpha}(x,W)-\alpha_0(x,W))(Y-g_0(x,W))\mid X=x\Rs\\
=~& \E\Ls (\hat{\alpha}(x,W)-\alpha_0(x,W))(\E[Y\mid X=x,W]-g_0(x,W))\mid X=x\Rs \\
=~& 0.
\end{align*}

\subsection{Proof of Proposition \ref{prop: dr}}
For the new score function $\psi(Z;\hat\alpha,\hat g) = m(Z;\hat g)+\hat \alpha(W)(Y-\hat g(W))$ with any fixed nuisance estimators $\hat{g}$ and $\hat{\alpha}$, we have
\begin{align*}
    &\E[\psi(Z;\hat\alpha,\hat g)\mid X=x]-\theta_0(x) \\
    =~& \E[m(Z;\hat g(x,W)) + \hat\alpha(x,W)(Y-\hat g(x,W))-\theta_0(x)\mid X=x]\\
    =~& \E[\alpha_0(x,W)\hat g(x,W) + \hat\alpha(x,W)(Y-\hat g(x,W))-\alpha_0(x,W)g_0(x,W)\mid X=x]\\
    =~& -\E[(\hat\alpha(x,W) - \alpha_0(x,W)) (\hat g(x,W)-g_0(x,W))\mid X=x],
\end{align*}
where the last equality holds since $\E[Y\mid W,X=x]=g_0(x,W)$ by definition.

\subsection{Proof of Theorem \ref{thm: consistency}}
The estimator $\hat{\theta}(x)$ employs the sub-sampled symmetric kernel
\[
K(x,X_i,Z_{1:n},\xi) = \frac{1}{B}\sum_{b=1}^B \tilde{K}(x,X_i,Z_{S_b},\xi_b),\quad i=1,\ldots,n,
\]
which satisfies the normalization condition that $\sum_{i=1}^n K(x,X_i,Z_{1:n},\xi)=1$. 
For notational simplicity, we denote the nuisance regression function and conditional Riesz representer estimated at point $(X_i, W_i)$ by $\hat{g}_i$ and $\hat{\alpha}_i$, respectively, i.e., $\hat{g}_i \coloneqq \hat{g}(X_i, W_i)$ and $\hat{\alpha}_i \coloneqq \hat{\alpha}(X_i, W_i)$.

Observe that we can decompose the estimation error $\hat\theta(x)-\theta_0(x)$ as
\begin{align*}
     ~&\hat{\theta}(x)-\theta_0(x) \\
     =~&\underbrace{\sum_{i=1}^n K(x,X_i,Z_{1:n},\xi)\Lp m(Z_i;\hat{g}_i)+\hat{\alpha}_i(Y_i-\hat{g}_i)\Rp  - \E_\xi\Ls \sum_{i=1}^n K(x,X_i,Z_{1:n},\xi)\Lp m(Z_i;\hat{g}_i)+\hat{\alpha}_i(Y_i-\hat{g}_i)\Rp\Rs}_{\text{subsampling error } E(\hat{g},\hat{\alpha}) }
     \\
     & + \underbrace{\E_\xi\Ls \sum_{i=1}^n K(x,X_i,Z_{1:n},\xi)\Lp m(Z_i;\hat{g}_i)+\hat{\alpha}_i(Y_i-\hat{g}_i)\Rp\Rs-\E\Ls \sum_{i=1}^n K(x,X_i,Z_{1:n},\xi)\Lp m(Z_i;\hat{g}_i)+\hat{\alpha}_i(Y_i-\hat{g}_i)\Rp\Rs}_{\text{sampling error }\Delta(\hat{g},\hat{\alpha}) }\\
     & + \underbrace{\E\Ls \sum_{i=1}^n K(x,X_i,Z_{1:n},\xi)\Lp m(Z_i;\hat{g}_i)+\hat{\alpha}_i(Y_i-\hat{g}_i)\Rp\Rs - \E\Ls m(Z_i;\hat{g}_i)+\hat{\alpha}_i(Y_i-\hat{g}_i)\mid X_i=x\Rs}_{\text{kernel error }\Gamma(\hat{g},\hat{\alpha}) }\\
     & + \underbrace{\E\Ls m(Z_i;\hat{g}_i)+\hat{\alpha}_i(Y_i-\hat{g}_i)\mid X_i=x\Rs - \E\Ls m(Z_i;g_0(x,W_i))\mid X_i=x\Rs}_{\text{nuisance error } \Lambda(\hat{g},\hat{\alpha}) }.
\end{align*}
Therefore, the estimation error can be bounded by: \[\La \hat{\theta}(x)-\theta_0(x)\Ra \leq \La E(\hat{g},\hat{\alpha})\Ra + \La\Delta(\hat{g},\hat{\alpha})\Ra + \La\Gamma(\hat{g},\hat{\alpha})\Ra + \La\Lambda(\hat{g},\hat{\alpha})\Ra.\]
In the sequel, we analyze the four terms $E(\hat{g},\hat{\alpha}), \Delta(\hat{g},\hat{\alpha}), \Gamma(\hat{g},\hat{\alpha})$, and $\Lambda(\hat{g},\hat{\alpha})$ respectively.

\begin{lemma}[Subsampling Error]  \label{lemma: subsampling error}
Under Assumption \ref{assump: moment}.\ref{subassump: boundedness}, w.p. $1-\delta$,
\[
 |E(\hat{g},\hat{\alpha})| = O\Lp \sqrt{\frac{\log(2/\delta)}{B}}\Rp.
\]
\end{lemma}
\begin{proof}
By the definition of $K(x,X_i,Z_{1:n},\xi)\,,$ we have 
\begin{align*}
    \hat{\theta}(x) =~& \sum_{i=1}^n K(x,X_i,Z_{1:n},\xi)\Ls m(Z_i;\hat{g}(X_i,W_i)) + \hat{\alpha}(X_i,W_i)(Y_i-\hat{g}(X_i,W_i))\Rs \\
    =~& \frac{1}{B}\sum_{b=1}^B \sum_{i\in S_b} \tilde{K}(x,X_i,Z_{S_b},\xi_b) \Ls m(Z_i;\hat{g}(X_i,W_i)) + \hat{\alpha}(X_i,W_i)(Y_i-\hat{g}(X_i,W_i))\Rs,
\end{align*}
where the second equality holds because $\tilde{K}(x,X_i,Z_{S_b},\xi_b) = 0$ for any $i \notin S_b$. We denote 
\[
V_b = \sum_{i\in S_b} \tilde{K}(x,X_i,Z_{S_b},\xi_b) \Ls m(Z_i;\hat{g}(X_i,W_i)) + \hat{\alpha}(X_i,W_i)(Y_i-\hat{g}(X_i,W_i))\Rs.
\]
Considering only the randomness in $V_b$ driven by $S_b$ and $\xi_b$, and by Assumption \ref{assump: moment}.\ref{subassump: boundedness}, $V_b$'s are bounded i.i.d. random variables. It follows from Hoeffding's bound that, for any $t\geq 0$, there exists a constant $c$ such that
\begin{align*}
    \PP\Lp \left| \frac{1}{B}\sum_{b=1}^B V_b - \E_\xi\Ls V_b\Rs\right| \geq t \Rp  \leq  2e^{-ct^2B}\,.
\end{align*}
Setting $t=\sqrt{\log(1/\delta)/(cB)}$ gives that, w.p. $1-\delta$, 
\[
|E(\hat{g},\hat{\alpha})| = \La \frac{1}{B}\sum_{b=1}^B V_b - \E_\xi[V_b]\Ra = O\Lp \sqrt{\frac{\log(2/\delta)}{B}}\Rp.
\]
\end{proof}

\begin{lemma}[Sampling Error]  \label{lemma: sampling error}
Under Assumption \ref{assump: moment}.\ref{subassump: boundedness}, w.p. $1-\delta$,
\[ |\Delta(\hat{g},\hat{\alpha})| = O\Lp \sqrt{\frac{s\log(2/\delta)}{n}}\Rp.
\]
\end{lemma}
\begin{proof}
Since we estimate $\hat{g}$ and $\hat{\alpha}$ based on the second half of the data $Z_{n+1:2n}$,
\begin{align*}
    & \E_\xi\Ls \sum_{i\in S_b} \tilde{K}(x,X_i,Z_{S_b},\xi_b) \Lp m(Z_i;\hat{g}(X_i,W_i)) + \hat{\alpha}(X_i,W_i)(Y_i-\hat{g}(X_i,W_i))\Rp\Rs \\
    =~&  \frac{1}{\binom{n}{s}} \sum_{S_b:|S_b|=s} \E_{\xi_b}\Ls \sum_{i\in S_b} \tilde{K}(x,X_i,Z_{S_b},\xi_b) \Lp m(Z_i;\hat{g}(X_i,W_i)) + \hat{\alpha}(X_i,W_i)(Y_i-\hat{g}(X_i,W_i))\Rp\Rs
\end{align*}
is a complete $U$-statistic with aggregation function
\[
f(Z_{S_b}) = \E_{\xi_b}\Ls \sum_{i\in S_b} \tilde{K}(x,X_i,Z_{S_b},\xi_b) \Lp m(Z_i;\hat{g}(X_i,W_i)) + \hat{\alpha}(X_i,W_i)(Y_i-\hat{g}(X_i,W_i))\Rp\Rs.
\]
The aggregation function $f(Z_{S_b})$ is symmetric and bounded. 
By Hoeffding's concentration inequality for $U$-statistics, we have w.p. $1-\delta$,
\[
 |\Delta(\hat{g},\hat{\alpha})| = \La \frac{1}{\binom{n}{s}} \sum_{S_b:|S_b|=s} f(Z_{S_b}) - \E\Ls \frac{1}{\binom{n}{s}} \sum_{S_b:|S_b|=s} f(Z_{S_b})\Rs \Ra = O\Lp \sqrt{\frac{s\log(2/\delta)}{n}}\Rp.
\]
\end{proof}

\begin{lemma}[Kernel Error]   \label{lemma: kernel error}
Under Assumptions \ref{assump: honesty}, \ref{assump: kernel shrinkage} and \ref{assump: moment}.\ref{subassump: Lip of moments}, we have $\Gamma(\hat{g},\hat{\alpha}) = O(\epsilon(s))$.
\end{lemma}
\begin{proof}
We let $Z_{S_b}^{(i)}=\{Z_j\}_{j\in S_b/\{i\}}$, $\hat{g}_i= \hat{g}(X_i,W_i)$ and $\hat{\alpha}_i= \hat{\alpha}(X_i,W_i)$. It holds that
\begin{align*}
&\E\Ls \sum_{i=1}^n K(x,X_i,Z_{1:n},\xi)(m(Z_i;\hat{g}_i)+\hat{\alpha}_i(Y_i-\hat{g}_i))\Rs \\
&=  \frac{1}{\binom{n}{s}} \sum_{S_b:|S_b|=s} \sum_{i\in S_b}\E\Ls  \tilde{K}(x,X_i,Z_{S_b},\xi_b) \Lp m(Z_i;\hat{g}_i) + \hat{\alpha}_i(Y_i-\hat{g}_i)\Rp\Rs \\
&=  \frac{1}{\binom{n}{s}} \sum_{S_b:|S_b|=s} \sum_{i\in S_b}\E\Ls\E\Ls  \tilde{K}(x,X_i,Z_{S_b},\xi_b) \Lp m(Z_i;\hat{g}_i) + \hat{\alpha}_i(Y_i-\hat{g}_i)\Rp\mid X_i, Z_{S_b}^{(i)}\Rs\Rs (\text{tower law})\\
&=  \frac{1}{\binom{n}{s}} \sum_{S_b:|S_b|=s} \sum_{i\in S_b}\E\Ls\E\Ls  \tilde{K}(x,X_i,Z_{S_b},\xi_b) \mid X_i, Z_{S_b}^{(i)}\Rs \cdot \E\Ls  m(Z_i;\hat{g}_i) + \hat{\alpha}_i(Y_i-\hat{g}_i) \mid X_i, Z_{S_b}^{(i)}\Rs\Rs (\text{honesty})\\
&=  \frac{1}{\binom{n}{s}} \sum_{S_b:|S_b|=s} \sum_{i\in S_b}\E\Ls\E\Ls  \tilde{K}(x,X_i,Z_{S_b},\xi_b) \mid X_i, Z_{S_b}^{(i)}\Rs \cdot \E\Ls  m(Z_i;\hat{g}_i) + \hat{\alpha}_i(Y_i-\hat{g}_i) \mid X_i\Rs\Rs (\text{i.i.d. data}) \\
&=  \frac{1}{\binom{n}{s}} \sum_{S_b:|S_b|=s} \sum_{i\in S_b}\E\Ls \tilde{K}(x,X_i,Z_{S_b},\xi_b) \cdot \E\Ls  m(Z_i;\hat{g}_i) + \hat{\alpha}_i(Y_i-\hat{g}_i) \mid X_i \Rs\Rs (\text{tower law}).
\end{align*}
We define $M(x) = \E\Ls  m(Z_i;\hat{g}_i) + \hat{\alpha}_i(Y_i-\hat{g}_i) \mid X_i=x \Rs$. Since $\sum_{i\in S_b}\tilde{K}(x,X_i,Z_{S_b},\xi_b) = 1$, the kernel error $\Gamma(\hat{g},\hat{\alpha})$ can be written as
\begin{align*}
\Gamma(\hat{g},\hat{\alpha}) =  \frac{1}{\binom{n}{s}} \sum_{S_b:|S_b|=s} \sum_{i\in S_b}\E\Ls \tilde{K}(x,X_i,Z_{S_b},\xi_b) ( M(X_i)-M(x))\Rs.
\end{align*}
By Assumption \ref{assump: moment}.\ref{subassump: Lip of moments}, $M(x)$ is $L$-Lipschitz in $x$. Therefore, we have
\begin{align*}
\sum_{i\in S_b}\E\Ls \tilde{K}(x,X_i,Z_{S_b},\xi_b) ( M(X_i)-M(x))\Rs &\leq L \cdot \sum_{i\in S_b}\E\Ls \tilde{K}(x,X_i,Z_{S_b},\xi_b) \cdot \Lv X_i - x\Rv\Rs \\
& \leq L \cdot\E\Ls\max_i\left\{\|X_i-x\|: \tilde{K}(x,X_i,Z_{S_b},\xi_b)>0\right\}\Rs\\
&\leq L\epsilon(s),
\end{align*}
where the last inequality follows from the kernel shrinkage in Assumption \ref{assump: kernel shrinkage}. Averaging over the subsamples $S_b$, we complete our proof.
\end{proof}

\begin{lemma}[Nuisance Error]  \label{lemma: nuisance error}
The nuisance error $\Lambda(\hat{g},\hat{\alpha})$ can be bounded by $|\Lambda(\hat{g},\hat{\alpha})| \leq \mathcal{E}(\hat{g},\hat{\alpha})$, where $\mathcal{E}(\hat{g},\hat{\alpha})= \|\hat{g}-g_0\|_x\cdot \|\hat{\alpha}-\alpha_0\|_x.$
\end{lemma}
\begin{proof}
By definition, $\E[m(Z_i;g(x,W_i))\mid X_i=x] = \E\Ls \alpha_0(x,W_i) g(x,W_i)\mid X_i=x\Rs$ for any $g\in \mathcal{G}_x\,.$ Plugging this into $\Lambda(\hat{g},\hat{\alpha})$, we have
\begin{align*}
    \Lambda(\hat{g},\hat{\alpha}) &= \E\Ls m(Z_i;\hat{g}_i)+\hat{\alpha}_i(Y_i-\hat{g}_i)\mid X_i=x\Rs - \E\Ls m(Z_i;g_0(x,W_i))\mid X_i=x\Rs \\
    & = \E\Ls \alpha_0(x,W_i)(\hat{g}(x,W_i)-g_0(x,W_i) ) + \hat{\alpha}(x,W_i)(Y_i-\hat{g}(x,W_i))\mid X_i=x\Rs.
\end{align*}
Further observe that
\begin{align*}
    \E\Ls \hat{\alpha}(x,W_i) Y_i\mid X_i=x\Rs &= \E\Ls \E\Ls\hat{\alpha}(x,W_i) Y_i\mid X_i=x,W_i\Rs\mid X_i=x\Rs \\
    & = \E\Ls \hat{\alpha}(x,W_i) \cdot\E\Ls Y_i\mid X_i=x,W_i\Rs\mid X_i=x\Rs  = \E\Ls \hat{\alpha}(x,W_i) g_0(x,W_i) \mid X_i=x\Rs.
\end{align*}
As a result,
\begin{align*}
    \Lambda(\hat{g},\hat{\alpha}) = \E\Ls \Lp\alpha_0(x,W_i)-\hat{\alpha}(x,W_i)\Rp(\hat{g}(x,W_i)-g_0(x,W_i))\mid X_i=x\Rs.
\end{align*}
By Hölder's inequality, we get that
\begin{align*}
    \La\Lambda(\hat{g},\hat{\alpha})\Ra &\leq \sqrt{\E\Ls \Lp\hat{g}(x,W_i)-g_0(x,W_i)\Rp^2\mid X_i=x\Rs\cdot \E\Ls \Lp\hat\alpha(x,W_i)-\alpha_0(x,W_i)\Rp^2\mid X_i=x\Rs}\\
    &= \|\hat{g}-g_0\|_x\cdot \|\hat{\alpha}-\alpha_0\|_x.
\end{align*}
\end{proof}

\vspace{2mm}
\noindent\textit{Proof of Theorem \ref{thm: consistency}:}
By Lemma \ref{lemma: subsampling error} and Lemma \ref{lemma: sampling error}, w.p. $1-2\delta$,
\[
\La E(\hat{g},\hat{\alpha})\Ra + \La \Delta(\hat{g},\hat{\alpha})\Ra =  O\Lp \sqrt{\frac{\log(2/\delta)}{B}} + \sqrt{\frac{s\log(2/\delta)}{n}}\Rp.
\]
Under the conditions that $s=o(n)$ and $B\geq n/s$, $\La E(\hat{g},\hat{\alpha})\Ra + \La \Delta(\hat{g},\hat{\alpha})\Ra=o_p(1).$ By Lemma \ref{lemma: kernel error} and Lemma \ref{lemma: nuisance error}, $\La \Gamma(\hat{g},\hat{\alpha}) \Ra + \La \Lambda(\hat{g},\hat{\alpha})\Ra=o(1)$ under the assumptions that $\epsilon(s)\rightarrow 0$ and $\mathcal{E}(\hat{g},\hat{\alpha})\rightarrow 0$. Combining these terms, we conclude that
\begin{align*}
    \La\hat{\theta}(x)-\theta_0(x)\Ra \leq~& \La E(\hat{g},\hat{\alpha})\Ra + \La \Delta(\hat{g},\hat{\alpha})\Ra + \La \Gamma(\hat{g},\hat{\alpha}) \Ra + \La \Lambda(\hat{g},\hat{\alpha})\Ra \\=~& o_p(1).
\end{align*}

\subsection{Proof of Corollary \ref{coro:consistency_DRRF}}
It suffices to show that the error rate of the nuisance estimators in \eqref{equ: lasso} of the DRRF algorithm satisfies $\mathcal{E}(\hat{g},\hat{\alpha})=O( n^{\frac{1}{2}(\beta-1)})$.

\begin{lemma}[\cite{orf}] \label{pro: nu guarantee}
Suppose Assumption \ref{assump: nuisance} holds, and $\hat{\nu}_g(x)$ and $\hat{\nu}_\alpha(x)$ are estimated as described in Section \ref{subsubsec: nusaince estimation}. Then, with probability at least $1-2\delta$, we have
\begin{align*}
    &\Lv \hat{\nu}_g(x)-\nu_g(x)\Rv_2 \leq \frac{2\lambda_g k}{\gamma-32k\sqrt{s\ln(d_{\nu}/\delta)/n}}\,, \\
    &\Lv \hat{\nu}_\alpha(x)-\nu_\alpha(x)\Rv_2 \leq \frac{2\lambda_\alpha k}{\gamma-32k\sqrt{s\ln(d_{\nu}/\delta)/n}}\,,
\end{align*}
as long as $\lambda_g,\lambda_\alpha \geq \Theta\Lp s^{-1/(2a d)}+\sqrt{\frac{s\ln(d_{\nu}/\delta)}{n}}\Rp.$
\end{lemma}

Since $g_0(x,w)=h(w,\nu_g(x))$ and $\alpha_0(x,w)=h(w,\nu_\alpha(x))$ for some known function $h$ that is $L$-Lipschitz in its second argument, we have for any $x\in \mathcal{X}$ and $w\in \mathcal{W}$,
\begin{align*}
&\Lp\hat{g}(x,w)-g_0(x,w)\Rp^2 = \Lp h(w,\hat\nu_g(x))-h(w,\nu_g(x))\Rp^2 \leq L^2\cdot \Lv \hat\nu_g(x) - \nu_g(x)\Rv^2, \\
&\Lp\hat{\alpha}(x,w)-\alpha_0(x,w)\Rp^2 = \Lp h(w,\hat\nu_\alpha(x))-h(w,\nu_\alpha(x))\Rp^2 \leq L^2\cdot \Lv \hat\nu_\alpha(x) - \nu_\alpha(x)\Rv^2.
\end{align*}
As a result,
\begin{align*}
    \mathcal{E}(\hat{g},\hat{\alpha}) \lesssim \Lv \hat\nu_g(x) - \nu_g(x)\Rv_2\cdot \Lv \hat\nu_\alpha(x) - \nu_\alpha(x)\Rv_2.
\end{align*}

Let $\lambda=\max(\lambda_g,\lambda_\alpha)$. According to Lemma \ref{pro: nu guarantee}, with probability $1-2\delta$, 
\begin{align*}
    &\Lv \hat{\nu}_g(x)-\nu_g(x)\Rv_2 \leq \frac{2\lambda_g k}{\gamma-32k\sqrt{s\ln(d_{\nu}/\delta)/n}}\,, \\
    &\Lv \hat{\nu}_\alpha(x)-\nu_\alpha(x)\Rv_2 \leq \frac{2\lambda_\alpha k}{\gamma-32k\sqrt{s\ln(d_{\nu}/\delta)/n}}\,,
\end{align*}
as long as $\lambda_g,\lambda_\alpha \geq \Theta\Lp s^{-1/(2a d)}+\sqrt{\frac{s\ln(d_{\nu}/\delta)}{n}}\Rp.$

We take $\delta = \delta_n = n^{\frac{1}{2}(\beta-1)}$, and $\lambda_g,\lambda_\alpha = \Theta\Lp s^{-1/(2a d)}+\sqrt{\frac{s\ln( d_\nu/\delta_n)}{n}}\Rp$. Note that $d_\nu$ grows at a polynomial rate in $n$. Therefore, $\gamma-32k\sqrt{s\ln(n\cdot 2d_\nu/\delta_n)/n} \rightarrow \gamma$ when $n\rightarrow \infty$. Furthermore, we have $s^{-1/(2ad)} \le \sqrt{\frac{s}{n}}$ as $\beta > (1+\frac{1}{ad})^{-1}$. Therefore, it holds that
\begin{align*}
    \mathcal{E}(\hat{g},\hat{\alpha}) \lesssim n^{\frac{1}{2}(\beta-1)}\Lp \frac{s}{n}\ln\Lp d_\nu \cdot n^{\frac{1}{2}(1-\beta)}\Rp+1\Rp = O( n^{\frac{1}{2}(\beta-1)}). 
\end{align*}

\subsection{Proof of Theorem \ref{thm: asymptotic normality}}

Note that the sampling error $$\Delta(\hat{g},\hat{\alpha})=\Delta(g_0,\alpha_0) + \Delta(\hat{g},\hat{\alpha}) - \Delta(g_0,\alpha_0).$$
Let $\hat{g}_i=\hat{g}(X_i,W_i), g_{0,i}=g_0(X_i,W_i)$, $\hat{\alpha}_i=\hat{\alpha}(X_i,W_i)$ and $\alpha_{0,i}=\alpha_0(X_i,W_i)$.
The last two terms $\Delta(\hat{g},\hat{\alpha}) - \Delta(g_0,\alpha_0)$ form a complete $U$-statistic with aggregation function 
\begin{align*}
f(Z_{S_b})=~&\E_{\xi_b} \Ls \sum_{i \in S_b}\tilde{K}(x,X_i,Z_{S_b},\xi_b)\Lp m(Z_i;\hat{g}_i) + \hat{\alpha}_i(Y_i-\hat{g}_i) - m(Z_i;g_{0,i}) - \alpha_{0,i}(Y_i-g_{0,i})\Rp\Rs\\&-\E \Ls \sum_{i \in S_b}\tilde{K}(x,X_i,Z_{S_b},\xi_b)\Lp m(Z_i;\hat{g}_i) + \hat{\alpha}_i(Y_i-\hat{g}_i) - m(Z_i;g_{0,i}) - \alpha_{0,i}(Y_i-g_{0,i})\Rp\Rs,
\end{align*}
which is symmetric and bounded.
By the Bernestein's inequality for $U$-statistics \citep{peel2010empirical}, we get that $$|\Delta(\hat{g},\hat{\alpha}) - \Delta(g_0,\alpha_0)| = O_p \Lp \sqrt{\frac{s\E[f(Z_{S_b})^2]}{n}}+\frac{s}{n}\Rp.$$

Note that 
\small
\begin{align*}
&\E\Ls f(Z_{S_b})^2\Rs\\ \le~&\E \Ls \sum_{i \in S_b}\tilde{K}(x,X_i,Z_{S_b},\xi_b)\Lp m(Z_i;\hat{g}_i) + \hat{\alpha}_i(Y_i-\hat{g}_i) - m(Z_i;g_{0,i}) - \alpha_{0,i}(Y_i-g_{0,i})\Rp^2\Rs (\text{Jensen's inequality})\\ =~&\E \Ls \sum_{i \in S_b}\E\Ls\tilde{K}(x,X_i,Z_{S_b},\xi_b)(m(Z_i;\hat{g}_i) + \hat{\alpha}_i(Y_i-\hat{g}_i) - m(Z_i;g_{0,i}) - \alpha_{0,i}(Y_i-g_{0,i}))^2\mid X_i,Z_{S_b}^{(i)}\Rs\Rs  (\text{tower law})\\ =~&\E \Ls \sum_{i \in S_b}\E\Ls\tilde{K}(x,X_i,Z_{S_b},\xi_b)\mid X_i,Z_{S_b}^{(i)}\Rs\E\Ls(m(Z_i;\hat{g}_i) + \hat{\alpha}_i(Y_i-\hat{g}_i) - m(Z_i;g_{0,i}) - \alpha_{0,i}(Y_i-g_{0,i}))^2\mid X_i,Z_{S_b}^{(i)}\Rs\Rs  (\text{honesty})\\ =~&\E \Ls \sum_{i \in S_b}\E\Ls\tilde{K}(x,X_i,Z_{S_b},\xi_b)\mid X_i,Z_{S_b}^{(i)}\Rs\E\Ls(m(Z_i;\hat{g}_i) + \hat{\alpha}_i(Y_i-\hat{g}_i) - m(Z_i;g_{0,i}) - \alpha_{0,i}(Y_i-g_{0,i}))^2\mid X_i\Rs\Rs (\text{i.i.d data})\\ =~& \E \Ls \sum_{i \in S_b}\tilde{K}(x,X_i,Z_{S_b},\xi_b)\E\Ls(m(Z_i;\hat{g}_i) + \hat{\alpha}_i(Y_i-\hat{g}_i) - m(Z_i;g_{0,i}) - \alpha_{0,i}(Y_i-g_{0,i}))^2\mid X_i\Rs\Rs (\text{tower law}).
\end{align*}
\normalsize
By the local mean-squared continuity condition in Assumption \ref{assump: msc},  
\begin{align*}
    &\E\Ls(m(Z_i;\hat{g}_i)+\hat{\alpha}_i(Y_i-\hat{g}_i) - m(Z_i;g_{0,i})-\alpha_{0,i}(Y_i-g_{0,i}))^2\mid X_i\Rs \\
    \leq~&  L\cdot \E\Ls(\hat{g}_i-g_{0,i})^2+(\hat{\alpha}_i-\alpha_{0,i})^2\mid X_i\Rs.
\end{align*}
Furthermore, since $g_0(x,w)=h(w,\nu_g(x))$ and $\alpha_0(x,w)=h(w,\nu_\alpha(x))$ with some known function $h$ that is $L$-Lipschitz in the parameter $\nu$, we have for any $x\in \mathcal{X}$ and $w\in \mathcal{W}$,
\begin{align*}
&\Lp\hat{g}(x,w)-g_0(x,w)\Rp^2 = \Lp h(w,\hat\nu_g(x))-h(w,\nu_g(x))\Rp^2 \leq L\cdot \Lv \hat\nu_g(x) - \nu_g(x)\Rv^2, \\
&\Lp\hat{\alpha}(x,w)-\alpha_0(x,w)\Rp^2 = \Lp h(w,\hat\nu_\alpha(x))-h(w,\nu_\alpha(x))\Rp^2 \leq L\cdot \Lv \hat\nu_\alpha(x) - \nu_\alpha(x)\Rv^2.
\end{align*}
As a result, $\E\Ls f(Z_{S_b})^2\Rs$ can be further bounded by
\begin{align*}
\E\Ls f(Z_{S_b})^2\Rs \lesssim~& \E \Ls \sum_{i \in S_b}\tilde{K}(x,X_i,Z_{S_b},\xi_b)\Lp \Lv \hat\nu_g(X_i) - \nu_g(X_i)\Rv^2 + \Lv \hat\nu_\alpha(X_i) - \nu_\alpha(X_i)\Rv^2\Rp\Rs\\
=~&\E\left[\E\left[\sum_{i \in S_b}\tilde{K}(x,X_i,Z_{S_b},\xi_b)\Lp \Lv \hat\nu_g(X_i) - \nu_g(X_i)\Rv^2 + \Lv \hat\nu_\alpha(X_i) - \nu_\alpha(X_i)\Rv^2\Rp\mid Z_{1:n}\right]\right].
\end{align*}

Note that we estimate ${\nu}_g$ and ${\nu}_\alpha$ using only the sample points $Z_{n+1},\ldots,Z_{2n}$, so $\hat{\nu}_g(\cdot)$ and $\hat{\nu}_\alpha(\cdot)$ are independent of $Z_i$ and $\tilde{K}(x,X_i,Z_{S_b},\xi_b)$ for $i=1,\ldots, n.$
Since for any fixed $x$, with probability at least $1-\delta_n$, where $\delta_n$ possibly depends on $n$,
$\Lp\|\hat{\nu}_g(x) - \nu_g(x)\|^2 + \|\hat{\nu}_\alpha(x) - \nu_\alpha(x)\|^2\Rp^{1/2}\leq r_n(\delta_n),$ we have
\begin{align*}
\E\left[\sum_{i \in S_b}\tilde{K}(x,X_i,Z_{S_b},\xi_b)\Lp \Lv \hat\nu_g(X_i) - \nu_g(X_i)\Rv^2 + \Lv \hat\nu_\alpha(X_i) - \nu_\alpha(X_i)\Rv^2\Rp\mid Z_{1:n}\right] \leq~ M\delta_n + r_n^2(\delta_n/n)\,
\end{align*}
for some constant $M$.

Therefore, $\E\Ls f(Z_{S_b})^2\Rs = O\Lp \delta_n + r_n^2(\delta_n/n)\Rp,$ and
 $$\Delta(\hat{g},\hat{\alpha}) - \Delta(g_0,\alpha_0)=O_p \Lp \sqrt{\frac{s\Lp \delta_n + r_n^2(\delta_n/n)\Rp}{n}}+\frac{s}{n}\Rp =O_p\Lp\frac{s}{n}+\delta_n+r_n^2(\delta_n/n)\Rp,
$$
where we have used $$\sqrt{\frac{s\Lp \delta_n + r_n^2(\delta_n/n)\Rp}{n}}\le \frac{1}{2}\Lp\frac{s}{n}+\delta_n+r_n^2(\delta_n/n)\Rp.$$

Note that the above bound dominates the subsampling error $E(\hat{g},\hat{\alpha})$ when $B \geq (n/s)^2$. In addition, the product error $\mathcal{E}(\hat{g},\hat{\alpha})$ can be bounded by
\[
\mathcal{E}(\hat{g},\hat{\alpha}) \leq L\cdot \Lv \hat\nu_g(x) - \nu_g(x)\Rv \cdot \Lv \hat\nu_\alpha(x) - \nu_\alpha(x)\Rv = O\Lp r_n^2(\delta_n/n)+\delta_n \Rp.
\]
Therefore, from the proof of Theorem \ref{thm: consistency}, we have that 
$$\La \hat{\theta}(x)-\theta_0(x)\Ra = \Delta(g_0,\alpha_0)+ O_p\Lp\frac{s}{n}+\delta_n+r_n^2(\delta_n/n)+\epsilon(s)\Rp,$$ for any $\delta_n > 0$ that may depend on $n$.

The leading term $\Delta(g_0,\alpha_0)$ is of the form $U_n-\E[U_n]$, where $U_n$ is  a complete $U$-statistic with aggregation function $$f(Z_{S_b})=\E_{\xi_b}\Ls \sum_{i\in S_b} \tilde K(x,X_i,Z_{S_b},\xi_b)\psi(Z_i;\alpha_0,g_0)\Rs.$$

To establish the asymptotic normality of $U_n$, we employ the following CLT result for $U$-statistics:

\begin{lemma}
Let $\rho(z)=\E[f(z,Z_2,...,Z_s)]$ and let $\eta(s)=Var(\rho(Z_1))$. Suppose that $n\eta(s)\rightarrow \infty$ and that $Var(f(Z_S))<\infty$. Let $U_n$ be the complete $U$-statistic of order $s$ associated with aggregation function $f$, then $$\sqrt{\frac{n}{s^2\eta(s)}}(U_n-\E[U_n]) \xrightarrow{d}N(0,1).$$
\end{lemma}
\begin{proof}
The proof follows identical steps as the one in Theorem 2 of \cite{fan2018dnn}.
\end{proof}

Since $f(Z_{S_b})=\E_{\xi_b}\Ls \sum_{i\in S_b} \tilde{K}(x,X_i,Z_{S_b},\xi_b)\psi(Z_i;\alpha_0,g_0)\Rs$ is uniformly bounded by Assumption \ref{assump: moment}.1, its variance is also bounded.
By Lemma 5, when $n\eta(s) \rightarrow \infty\,,$ we have that  $$\sqrt{\frac{n}{s^2\eta(s)}}\Delta(g_0,\alpha_0) \xrightarrow{d}N(0,1)\,.$$

Note that under the following four conditions: (1) $n\eta(s) \rightarrow \infty$; (2) $\sqrt{n/(s^2\eta(s))}\cdot \delta_n \rightarrow 0$; (3) $\sqrt{n/(s^2\eta(s))}\cdot r_n^2(\delta_n/n) \rightarrow 0$; (4) $\sqrt{n/(s^2\eta(s))}\cdot \epsilon(s)\rightarrow0$,
we have that $$\sqrt{\frac{n}{s^2\eta(s)}}\cdot O_p\Lp\frac{s}{n}+\delta_n+r_n^2(\delta_n/n)+\epsilon(s)\Rp=o_p(1),$$ from which Slutsky's lemme gives $$\sqrt{\frac{n}{s^2\eta(s)}}(\hat{\theta}(x)-\theta_0(x)) \xrightarrow{d}N(0,1).$$

\vspace{3mm}
\subsection{Proof of Theorem \ref{thm: asymptotic normality for DRRF}}

We first show that the two conditions $n\eta(s) \rightarrow \infty$ and $n\epsilon^2(s)/(s^2\eta(s))\rightarrow 0$ in Theorem \ref{thm: asymptotic normality} are satisfied under the doubly robust random forest algorithm described in Section \ref{subsec: drrf}.

\begin{proposition} \label{pro: rate conditions}
    Suppose that the base kernel $\tilde K$ is constructed from a DRRF that satisfies Assumption \ref{assump: forest}. Additionally, suppose there exists a strictly positive $\underline\sigma$ such that $min_x Var(\psi(Z;\alpha_0,g_0)\mid X=x)\ge\underline{\sigma}^2>0\,.$ Set $s=\Theta(n^{\beta})$, where $\Lp 1+\frac{1}{\alpha d}\Rp^{-1}<\beta<1$. Then, $\gamma(s)=w\Lp\epsilon(s)/s\Rp$, and the conditions in Theorem \ref{thm: asymptotic normality} that $n\eta(s) \rightarrow \infty$ and $n\epsilon^2(s)/(s^2\eta(s))\rightarrow 0$ hold.
\end{proposition}

We first prove Proposition \ref{pro: rate conditions}.
Note that when the kernel $\tilde{K}$ is obtained from a DRRF that satisfies Assumption \ref{assump: forest}, \cite{wager2018estimation} shows that $\epsilon(s)=O(s^{-\frac{1}{2ad}})$, where $a$ is a positive constant that depends on the parameters specified in Assumption \ref{assump: forest}. Moreover, Lemma 4 of \cite{wager2018estimation} proves that $\gamma(s)\gtrsim \frac{1}{s\log(s)^d}$. Hence, $\gamma(s)=w\Lp\epsilon(s)/s\Rp$.

To prove the two conditions, $n\eta(s) \rightarrow \infty$ and $n\epsilon^2(s)/(s^2\eta(s))\rightarrow 0$, hold, we first formalize the statement that $\eta(s)$ can be lower bounded by a constant multiple of $\gamma(s)$ in Lemma \ref{lemma: eta and gamma}. To prove Lemma \ref{lemma: eta and gamma}, we begin with Lemma \ref{lemma:strong_honesty}.

\begin{lemma}  \label{lemma:strong_honesty}
Suppose $\{(K_i,H_i,X_i)\}_{i=1}^s$ is a set of random variables.
Assume the following conditions hold: 
\[
K_1 \perp \!\!\! \perp H_1 \mid X_1, \] and $$\mathbb{E}[K_i H_i \mid H_1, X_1] = \mathbb{E}[K_i H_i \mid X_1]\quad \text{for} \quad i > 1. $$
Additionally, suppose there exists a strictly positive $\sigma$ such that $min_x Var(H_1\mid X_1=x)\ge {\sigma}^2>0$. Let \( T := \sum_{i=1}^s K_i H_i \). Furthermore, assume  $\E[T\mid Z_1]=\E[T\mid H_1,X_1]$. Then,
\[
\Var \Lp \E\Ls T\mid Z_1 \Rs\Rp \geq \sigma^2 \cdot\E\Ls \E[K_1\mid X_1]^2\Rs.
\]

\end{lemma}
\begin{proof}
We have the decomposition
\[
\mathbb{E}[T \mid Z_1] = \underbrace{\mathbb{E}[T \mid X_1]}_{A} + \underbrace{\mathbb{E}[T \mid H_1, X_1] - \mathbb{E}[T \mid X_1]}_{B}.
\]
Note that \( A \) and \( B \) are uncorrelated. Thus, we have
\[
\Var\left(\mathbb{E}[T \mid Z_1]\right) \geq \text{Var}(B) = \text{Var} \left( \sum_{i=1}^s \left( \mathbb{E}[K_i H_i \mid H_1, X_1] - \mathbb{E}[K_i H_i \mid X_1] \right) \right).
\]

Since for all \( i > 1 \),
$\mathbb{E}[K_i H_i \mid H_1, X_1] = \mathbb{E}[K_i H_i \mid X_1]$, all terms in the summation, except the one associated with \( i = 1 \), vanish. Therefore,
\begin{align*}
    \Var\left(\mathbb{E}[T \mid Z_1]\right) \geq~& \mathbb{E}\left[\left(\E[K_1 H_1 \mid H_1, X_1] - \mathbb{E}[K_1 H_1 \mid X_1]\right)^2\right] \\
    =~& \mathbb{E}\Ls\left(\mathbb{E}[K_1 \mid H_1, X_1]H_1 - \mathbb{E}[K_1 \mid X_1]\mathbb{E}[H_1 \mid X_1]\right)^2\Rs \quad \text{(honesty)}\\
    =~& \mathbb{E}\left[\Lp\mathbb{E}[K_1 \mid  X_1]H_1 - \mathbb{E}[K_1 \mid X_1] \mathbb{E}[H_1 \mid X_1]\Rp^2\right] \quad \text{(honesty)}\\
    =~& \mathbb{E}\left[\mathbb{E}[K_1 \mid X_1]^2\mathbb{E}\left[\left(H_1 - \mathbb{E}[H_1 \mid X_1]\right)^2 \mid X_1\right]\right] \quad \text{(tower law)} \\
    \geq~& \sigma^2 \mathbb{E}[\mathbb{E}[K_1 \mid X_1]^2].
\end{align*}
\end{proof}

\begin{lemma}   \label{lemma: eta and gamma} 
Suppose that the DRRF constructed satisfies Assumption \ref{assump: forest}. Furthermore, there exists a strictly positive $\underline\sigma$ such that $min_x Var(\psi(Z;\alpha_0,g_0)\mid X=x)\ge\underline{\sigma}^2>0\,.$ Then, for sufficiently large $s$, there exists a constant $c>0$ such that $$\eta(s) \geq c\cdot\gamma(s)\,.$$
\end{lemma}
\begin{proof}
The proof closely follows the argument in \cite{wager2018estimation}. With the doubly robust random forest algorithm in Section \ref{subsec: drrf}, we partition the samples in $S$ into two equal-sized subsets $S^1$ and $S^2$. The kernel is constructed using only the sample points from $S^1$, and the weight assigned to all samples in $S^1$ is set to zero. For the samples in $S^2$, the kernel construction satisfies strong honesty, i.e.: $$\forall i,j \in S^2: \tilde K(x,X_i,Z_{S_b},\xi_b) \indep \psi(Z_j;\alpha_0,g_0)\mid X_j\,.$$
To simply notation, let
\[
T = \sum_{i=1}^s\tilde{K}(x,X_i,Z_S,\xi_S)\psi(Z_i;\alpha_0,g_0).
\]
We have $\mathbb{E}[T \mid Z_1] = \mathbb{E}[T \cdot I\{1 \in S^1\} \mid Z_1] + \mathbb{E}[T\cdot I\{1 \in S^2\} \mid Z_1].$ Since \(\text{Var}(A+B) -(\frac{1}{2}\text{Var}(A) - \text{Var}(B))=2\Var(B+\frac{1}{2}A)\geq 0\),
\begin{align*}
    \eta(s)=\text{Var}(\mathbb{E}[T \mid Z_1]) \geq~& \frac{1}{2}\text{Var}(\mathbb{E}[T \mid Z_1, 1 \in S^2]\PP(1 \in S^2)) - \text{Var}(\mathbb{E}[T \mid Z_1, 1 \in S^1]\PP(1 \in S^1)) \\
    =~& \frac{1}{8}\text{Var}(\mathbb{E}[T \mid Z_1, 1 \in S^2]) - \frac{1}{4}\text{Var}(\mathbb{E}[T \mid Z_1, 1 \in S^1]).
\end{align*}
The last equality holds because $S^1$ and $S^2$ are randomly selected from $S$ with equal probability.

Moreover, conditioning on the event \( \{1 \in S^2\} \) doesn’t change the distribution of the random variables, since \( S^2 \) is chosen fully at random without looking at the data. It is straightforward to verify that the assumptions of Lemma \ref{lemma:strong_honesty} hold for $K_i=\tilde{K}(x,X_i,Z_S,\xi_S)$, and $H_i=\psi(Z_i;\alpha_0,g_0)$, $X_i=X_i$ under $\{1 \in S^2\}$. Therefore,
\[
\text{Var}(\mathbb{E}[T \mid Z_1, 1 \in S^2]) \geq \underline\sigma^2 \gamma(s).
\]

We now argue that the negative term is negligible. Let \( i_1 \) denote the random variable of the first index of the set \( S^1 \), i.e. \( S^1 = \{i_1, \dots, i_{s/2}\} \). Note that since \( T \) is stochastically symmetric under permutations of the input variables, we can write
\[
\mathbb{E}[T \mid Z_1, 1 \in S^1] = \mathbb{E}[T \mid Z_{i_1}].
\]
Note that the elements of \( S^1 \) can be thought of as a random subset of size \( s/2 \) among the \( n \) samples. Moreover, consider the symmetric function
\[
g(Z_{S^1}) = \mathbb{E}[T \mid \{Z_i\}_{i \in S^1}].
\]
We can view this as the aggregation function of a \( U \)-statistic. In the following, we perform an Analysis of Variance (ANOVA) decomposition as done in \cite{fan2018dnn}.

Consider the following projection functions:
\begin{align*}
    g_1(z_{i_1}) &= \mathbb{E}[g(z_{i_1}, Z_{i_2}, \dots, Z_{i_{s/2}})], \\
    g_2(z_{i_1}, z_{i_2}) &= \mathbb{E}[g(z_{i_1}, z_{i_2}, Z_{i_3}, \dots, Z_{i_{s/2}})], \\
    &\vdots \\
    g_{s/2}(z_{i_1}, z_{i_2}, \dots, z_{i_{s/2}}) &= g(z_{i_1}, z_{i_2}, \dots, z_{i_{s/2}}),
\end{align*}
and
\begin{align*}
    \tilde{g}_1(z_{i_1}) &= g_1(z_{i_1}) - \mathbb{E}[g], \\
    \tilde{g}_2(z_{i_1}, z_{i_2}) &= g_2(z_{i_1}, z_{i_2}) - \mathbb{E}[g], \\
    &\vdots \\
    \tilde{g}_{s/2}(z_{i_1}, z_{i_2}, \dots, z_{i_{s/2}}) &= g_{s/2}(z_{i_1}, z_{i_2}, \dots, z_{i_{s/2}}) - \mathbb{E}[g],
\end{align*}
where \( \mathbb{E}[g] = \mathbb{E}[g(Z_{S^1})] \). Then we define the canonical terms of Hoeffding's $U$-statistic decomposition as
\begin{align*}
    h_1(z_{i_1}) &= \tilde{g}_1(z_{i_1}), \\
    h_2(z_{i_1}, z_{i_2}) &= \tilde{g}_2(z_{i_1}, z_{i_2}) - h_1(z_{i_1}) - h_1(z_{i_2}), \\
    h_3(z_{i_1}, z_{i_2}, z_{i_3}) &= \tilde{g}_3(z_{i_1}, z_{i_2}, z_{i_3}) - \sum_{k=1}^3 h_1(z_{i_k}) - \sum_{1 \leq k < l \leq 3} h_2(z_{i_k}, z_{i_l}), \\
    &\vdots \\
    h_{s/2}(z_{i_1}, z_{i_2}, \dots, z_{i_{s/2}}) &= \tilde{g}_{s/2}(z_{i_1}, z_{i_2}, \dots, z_{i_{s/2}}) - \sum_{k=1}^{s/2} h_1(z_{i_k}) - \sum_{1 \leq k < l \leq s/2} h_2(z_{i_k}, z_{i_l}) - \cdots  \\
    &\quad -\sum_{1 \leq k_1 < k_2 < \cdots < k_{(s/2)-1} \leq s/2} h_{s/2-1}(z_{i_{k_1}}, z_{i_{k_2}}, \dots, z_{i_{k_{(s/2)-1}}}).
\end{align*}
Subsequently, 
\begin{align*}
    \tilde{g}_{s/2}(z_{i_1}, \dots, z_{i_{s/2}}) &= g(z_{i_1}, \dots, z_{i_{s/2}}) - \mathbb{E}[g]\\ &= \sum_{k=1}^{s/2} h_1(z_{i_k}) + \sum_{1 \leq k < l \leq s/2} h_2(z_{i_k}, z_{i_l}) + \cdots + h_{s/2}(z_{i_1}, \dots, z_{i_{s/2}}).
\end{align*}

It is straightforward to show that all the canonical terms in the latter expression are uncorrelated. Hence, we have
\begin{align*}
    \text{Var}[g(Z_{i_1}, \dots, Z_{i_{s/2}})] &= \binom{s/2}{1} \mathbb{E}[h_1^2] + \binom{s/2}{2} \mathbb{E}[h_2^2] + \cdots + \binom{s/2}{s/2} \mathbb{E}[h_{s/2}^2] \\
    &\geq \frac{s}{2} \mathbb{E}[h_1^2] = \frac{s}{2} \text{Var}(\mathbb{E}[T \mid Z_{i_1}]),
\end{align*}
which implies $\text{Var}(\mathbb{E}[T \mid Z_{i_1}]) \leq 2/s \cdot\text{Var}(\mathbb{E}[T \mid \{Z_i\}_{i \in S^1}]).$
Now observe that
\begin{align*}
    \mathbb{E}[T \mid \{Z_i\}_{i \in S^1}]
    =~& \mathbb{E}\left[\sum_{i\in S^2}\tilde{K}(x,X_i,Z_S,\xi_S)\psi(Z_i;\alpha_0,g_0) \mid \{Z_i\}_{i \in S^1}\right] \\
    =~& \mathbb{E}\left[\sum_{i \in S^2}\tilde{K}(x,X_i,Z_S,\xi_S) \mathbb{E}[\psi(Z_i;\alpha_0,g_0) \mid X_i] \mid \{Z_i\}_{i \in S^1}\right] \\
    =~& \mathbb{E}[\psi(Z_i;\alpha_0,g_0) \mid X = x]+ \\
    &\mathbb{E}\Ls\sum_{i \in S^2} \tilde{K}(x,X_i,Z_S,\xi_S)(\mathbb{E}[\psi(Z_i;\alpha_0,g_0) \mid X_i] - \mathbb{E}[\psi(Z;\alpha_0,g_0) \mid X = x]) \mid \{Z_i\}_{i \in S^1}\Rs.
\end{align*}
Therefore,
\begin{align*}
    &\text{Var}(\mathbb{E}[T \mid \{Z_i\}_{i \in S^1}]) \\
    \leq~& \mathbb{E} \left[ \mathbb{E} \left[ \sum_{i \in S^2} \tilde{K}(x,X_i,Z_S,\xi_S) (\mathbb{E}[\psi(Z_i;\alpha_0,g_0) \mid X_i] - \mathbb{E}[\psi(Z;\alpha_0,g_0) \mid X = x]) \mid \{Z_i\}_{i \in S^1} \right]^2 \right] \\
    \leq~&L^2 \mathbb{E} \left[ \mathbb{E} \left[ \max_i\{ \|X_i - x\|, i \in S^2 : \tilde{K}(x,X_i,Z_S,\xi_S) > 0 \} \mid \{Z_i\}_{i \in S^1} \right]^2 \right].
\end{align*}

Note that the kernel shrinkage property holds for any fixed values of the partition samples for honest forests that satisfy the balancedness and minimal random splitting criteria. Therefore, we have
\[
\text{Var}(\mathbb{E}[T \mid \{Z_i\}_{i \in S^1}])  = O(\epsilon(s)^2).
\]
Thus, overall, we have
\[
\eta(s)=\text{Var}(\mathbb{E}[T \mid Z_1]) \gtrsim \underline\sigma^2 \gamma(s) - O\Lp\frac{\epsilon(s)^2}{s}\Rp.
\]

Since $\gamma(s)=w\Lp\epsilon(s)/s\Rp$, we have $\frac{\epsilon(s)^2}{s}=o(\epsilon(s)\gamma(s))$ is of lower order than $\gamma(s)$. 
\end{proof}

\vspace{1mm}

 If we choose $s=\Theta(n^{\beta})$, where $\Lp 1+\frac{1}{ad}\Rp^{-1}<\beta<1$, then $\frac{n}{s\log(s)^d}\rightarrow \infty$ and $ns^{-(1+\frac{1}{ad}-\epsilon)} \rightarrow 0$ for some $\epsilon > 0$, which implies that $n\gamma(s) \rightarrow \infty$ and $n\epsilon^2(s)/(s^2\gamma(s))\rightarrow 0$.
By Lemma \ref{lemma: eta and gamma}, $\eta(s) \geq c\cdot\gamma(s)$ for sufficiently large $s$. Therefore, $n\eta(s) \rightarrow \infty$ and $n\epsilon^2(s)/(s^2\eta(s))\rightarrow 0$, and we complete the proof of Proposition \ref{pro: rate conditions}.

\vspace{2mm}
In the following, we prove that the nuisance conditions $$\sqrt{n/(s^2\eta(s))}\cdot \delta_n \rightarrow 0 \quad \text{and}\ \ \sqrt{n/(s^2\eta(s))}\cdot r_n^2(\delta_n/n)\rightarrow 0$$ from Theorem \ref{thm: asymptotic normality} are satisfied with $\delta_n$ appropriately chosen for the DRRF algorithm.

Let $\lambda=\max(\lambda_g,\lambda_\alpha)$. According to Proposition \ref{pro: nu guarantee},
\begin{align*}
r_n(\delta)= O\Lp\frac{2\lambda k}{\gamma-32k\sqrt{s\ln(2 d_\nu/\delta)/n}}\Rp.
\end{align*}
We take $\delta = \delta_n = n^{-c}$, where $c>1/2$, and $\lambda_g,\lambda_\alpha = \Theta\Lp s^{-1/(2ad)}+\sqrt{\frac{s\ln( d_\nu/\delta_n)}{n}}\Rp$. Note that $d_\nu$ grows at a polynomial rate in $n$. Therefore, $$\gamma-32k\sqrt{s\ln(n\cdot 2d_\nu/\delta_n)/n} \rightarrow \gamma$$ when $n\rightarrow \infty$. 
Hence, we only need to show $$n^{-c}\sqrt{n/(s^2\eta(s))} \rightarrow 0 \quad \text{and}\quad \frac{s}{n}\ln(n^{1+c}\cdot d_\nu) \cdot \sqrt{n/(s^2\eta(s))}\rightarrow 0,$$ where the second condition is sufficient since $s^{-1/(2ad)} \le \sqrt{\frac{s}{n}}$ as $\beta > (1+\frac{1}{ad})^{-1}$. In the case of doubly robust random forest, $\gamma(s)\gtrsim \frac{1}{s\log(s)^d}$. This completes the proof.

\section{Additional Simulation Results}  \label{appendix:simu}
This section presents additional simulation results for the continuous treatment setting, following Section \ref{sec:simulations}.  We focus on the comparison between DRRF variants and ORF variants \citep{orf}, as GRF variants are not applicable in this context. Similar to the binary treatment setting in Section \ref{sec:simulations}, we consider setups with various data-generating processes for the treatment variable $D_i$ and additional covariates $\tilde X_i$. The detailed description of the data-generating processes is provided in Table \ref{tab:res_cts}.

Table \ref{tab:res_cts} compares the estimation accuracy and evaluation runtime of different estimators for various continuous treatment setups. Specifically, it reports the root mean squared error (RMSE) of the estimated $\hat{\theta}(x_i)$ and the average evaluation runtime across 100 test points $x_i$ over a uniform grid in $[0,1]$. As before, we report results from 100 simulations for each setup. Figures \ref{fig:cts_1}-\ref{fig:cts_3} illustrate the performance of each estimator by plotting the average estimated effects, along with the 5th and 95th percentiles of the estimated effects across 100 replications.

\begin{table}[htbp]
\centering
\begin{tabular}{c|ccc|ccc}
\toprule
\multicolumn{1}{c}{} & \multicolumn{3}{c|}{RMSE} &\multicolumn{3}{c}{Time (seconds)}\\
\midrule
\multicolumn{1}{c}{} &\multicolumn{6}{c}{\textbf{Setup 1}}\\
\cmidrule(r){2-7}
\multicolumn{1}{c}{} &$k=5$ &$k=10$ &$k=15$  &$k=5$ &$k=10$ &$k=15$ \\
\midrule
 DRRF &0.121 &0.105 &0.091 &0.16 &0.17 &0.17\\
 DRRF-CV &0.120 &0.105 &0.093 &0.17 &0.17 &0.17\\
\midrule
 ORF &0.125 &0.115 &0.120 &0.55 &0.54 &0.53 \\
 ORF-CV &0.090 &0.101 &0.106 &17.8 &17.9 &18.5\\
\midrule
\multicolumn{1}{c}{} &\multicolumn{6}{c}{\textbf{Setup 2}}\\
\cmidrule(r){2-7}
\multicolumn{1}{c}{} &$k=5$ &$k=10$ &$k=15$  &$k=5$ &$k=10$ &$k=15$ \\
\midrule
DRRF &0.057 &0.082 &0.112 &0.17 &0.16 &0.17\\
DRRF-CV &0.049 &0.038 &0.039&0.17 &0.17 &0.17\\
\midrule
 ORF &0.127 &0.194 &0.259 &1.5 &1.6 &1.6\\
 ORF-CV &0.096 &0.131 &0.162 &25.2 &24.2 &24.1\\
\midrule
\multicolumn{1}{c}{} &\multicolumn{6}{c}{\textbf{Setup 3}}\\
\cmidrule(r){2-7}
\multicolumn{1}{c}{} &$k=5$ &$k=10$ &$k=15$  &$k=5$ &$k=10$ &$k=15$ \\
\midrule
 DRRF &0.062 &0.088 &0.142 &0.16 &0.17 &0.17\\
 DRRF-CV &0.059 &0.036 &0.034&0.17 &0.17 &0.17\\
\midrule
 ORF &0.161 &0.177 &0.191 &3.1 &3.4 &3.5\\
 ORF-CV &0.150 &0.242 &0.332 &25.4 &25.3 &26.2\\
\bottomrule
\end{tabular}
\caption{RMSE and runtime for continuous treatment setups}
\floatfoot{This table reports the root mean squared error (RMSE) and the evaluation runtime for different estimators in continuous treatment setups. We generate data from the process $Y_i=\theta_0(X_i)\cdot D_i+\tilde X_i^\top \nu_0+\epsilon_i$, where $X_i\overset{i.i.d.}{\sim} U[0,1]$ and $\nu_0$ is a $k$-sparse vector with non-zero coefficients drawn i.i.d. from $U[-1,1]$. The continuous treatment variables $D_i$'s and the additional covariates $\tilde X_i$ are generated as follows: (1) Setup 1: $D_i = \tilde{X}_i^\top \gamma_0 + \eta_i$, and $\tilde{X}_i \sim \mathcal{N}(0, I_p)$. (2) Setup 2: $D_i = (\tilde{X}_i^\top \gamma_0)^3 + \eta_i$, and $\tilde{X}_i \sim \mathcal{N}(2X_i, I_p)$. (3) Setup 3:  $D_i = (\tilde{X}_i^\top \gamma_0)^2 + \eta_i$, and $\tilde{X}_i \sim \mathcal{N}(2, I_p)$. In each setup, $\gamma_0$ is a $k$-sparse vector with the same support as $\nu_0$, where its non-zero coefficients are drawn i.i.d. from $U[0,1]$, and the noise terms are drawn as $\eta_i \overset{i.i.d.}{\sim} U[-1,1]$. We set the number of covariates $\tilde X_i$ to be $p=100$ and the sample size $2n=5000$. We perform 100 simulations for each setup. The evaluation runtime is the average runtime over 10 runs for estimating $\hat{\theta}(x_i)$ at 100 test points $x_i$ over a uniform grid in $[0,1]$. The results are measured on a 2019 MacBook Pro with a 2.6 GHz 6-Core Intel Core i7.}
\label{tab:res_cts}
\end{table}

We confirm our findings from the binary treatment setting, showing that ORF variants can suffer from significant bias or variance when the model for the treatment variable is misspecified. As shown in Table \ref{tab:res_cts}, DRRF variants outperform ORF variants in most cases in terms of RMSE when the relationship between treatment $D_i$ and covariates $\tilde{X}_i$ is non-linear. In the first continuous treatment setup, which replicates a setup in \cite{orf} with a correctly specified model for residualization (linear relationship between treatment $D_i$ and covariates $\tilde{X}_i$), our DRRF algorithms perform comparably to the ORF variants. However, in the second and third continuous setups, where the relationship between treatment $D_i$ and covariates $\tilde{X}_i$ is non-linear, the DRRF variants significantly outperform both ORF and ORF-CV. Specifically, in the second continuous treatment setup (as shown in Figure \ref{fig:cts_2}), ORF and ORF-CV yield biased estimates, particularly for larger values of $k$. In the third continuous treatment setup (as shown in Figure \ref{fig:cts_3}), ORF and ORF-CV show substantial variance in estimating the true effects. Additionally, DRRF variants are significantly more computationally efficient than ORF variants, especially ORF-CV, when estimating $\theta_0(x)$ at multiple query points. As shown in Table \ref{tab:res_cts}, the evaluation runtimes of both DRRF and DRRF-CV are significantly shorter than those of ORF and ORF-CV.

\begin{figure}[t!]
    \centering
    \begin{subfigure}{\textwidth}
        \centering
        \includegraphics[width=0.8\linewidth]{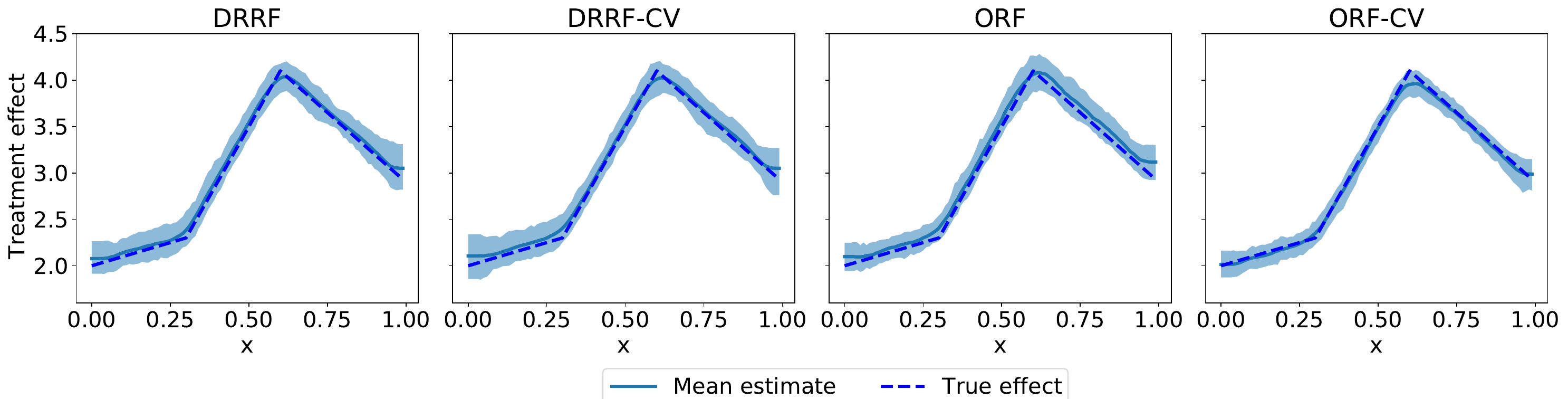} 
        \caption{$k=5$}
    \end{subfigure}
    \begin{subfigure}{\textwidth}
        \centering
        \includegraphics[width=0.8\linewidth]{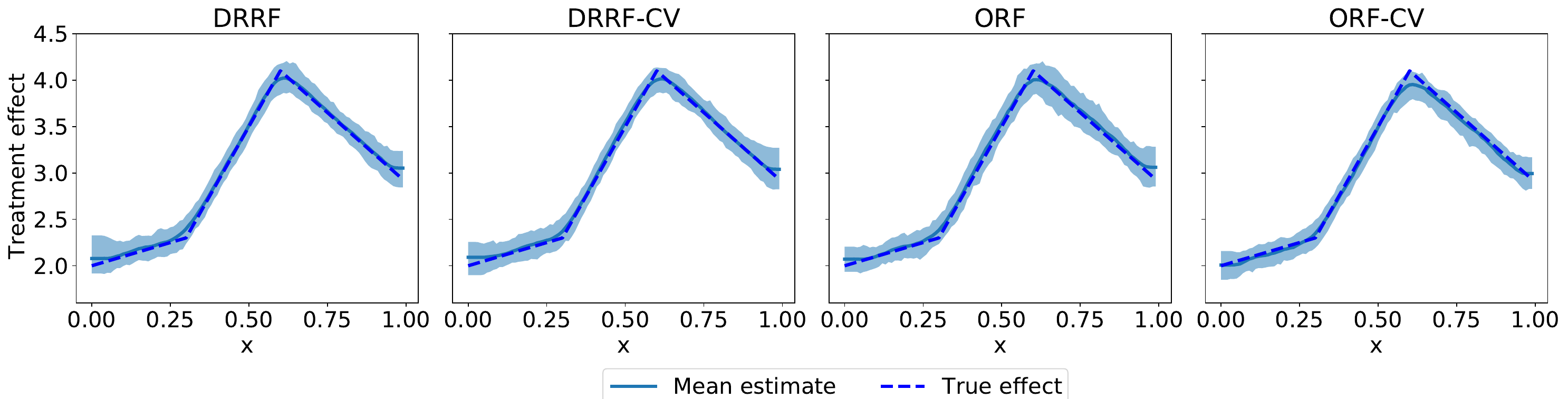} 
        \caption{$k=10$}
    \end{subfigure}
    \begin{subfigure}{\textwidth}
        \centering
        \includegraphics[width=0.8\linewidth]{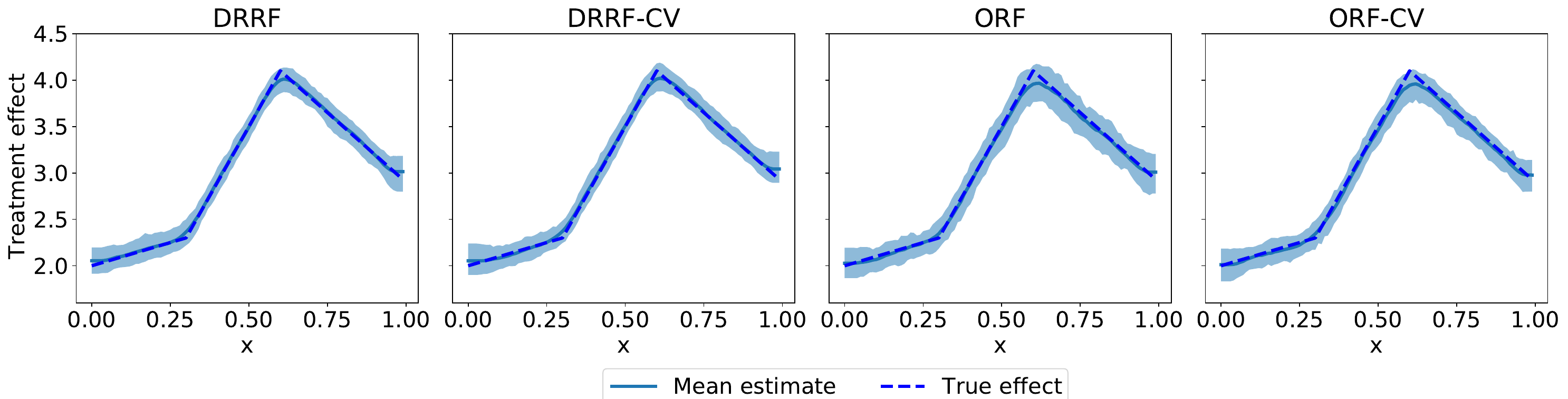} 
        \caption{$k=15$}
    \end{subfigure}
    \caption{Treatment effect estimations for Setup 1 in continuous treatment setting}
    \floatfoot{This figure shows the treatment effect estimations for Setup 1 in the continuous treatment setting. The plot shows the average estimated effects, with shaded regions representing the 5th to 95th percentile intervals of the estimated effects across 100 replications. The data is generated from the process $Y_i=\theta_0(X_i)\cdot D_i+\tilde X_i^\top \nu_0+\epsilon_i$, where $X_i\overset{i.i.d.}{\sim} U[0,1]$, $\tilde{X}_i \sim \mathcal{N}(0, I_p)$, and $D_i = \tilde{X}_i^\top \gamma_0 + \eta_i$. Both $\nu_0$ and $\gamma_0$ are $k$-sparse vectors with the same support, and their non-zero coefficients are drawn i.i.d. from $U[-1,1]$. We set the number of observations to $2n=5000$ and $p=100$.}
    \label{fig:cts_1}
\end{figure}

\begin{figure}[t!]
    \centering
    \begin{subfigure}{\textwidth}
        \centering
        \includegraphics[width=0.8\linewidth]{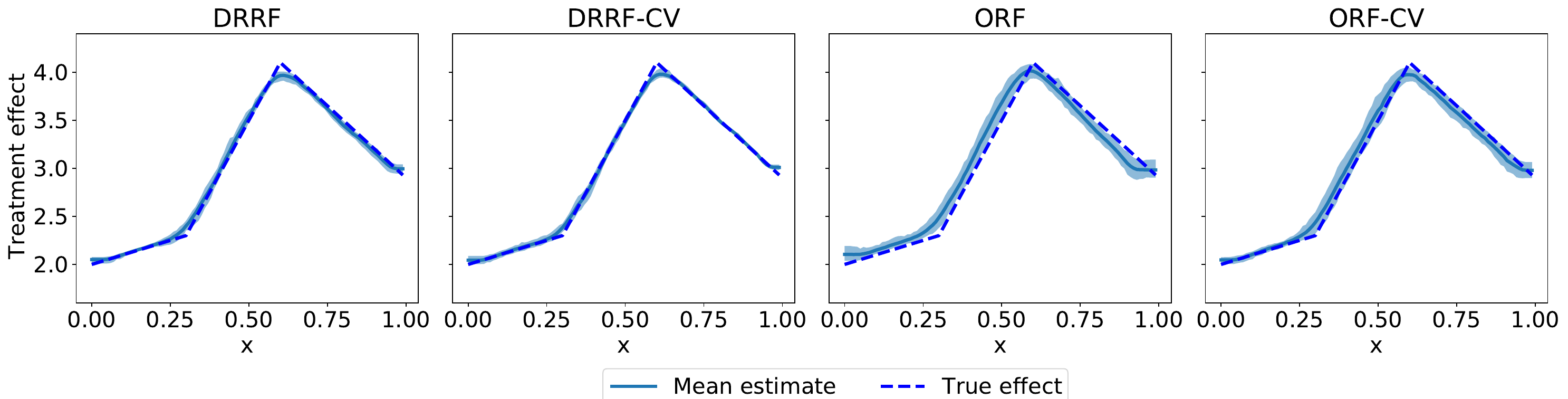} 
        \caption{$k=5$}
    \end{subfigure}
    \begin{subfigure}{\textwidth}
        \centering
        \includegraphics[width=0.8\linewidth]{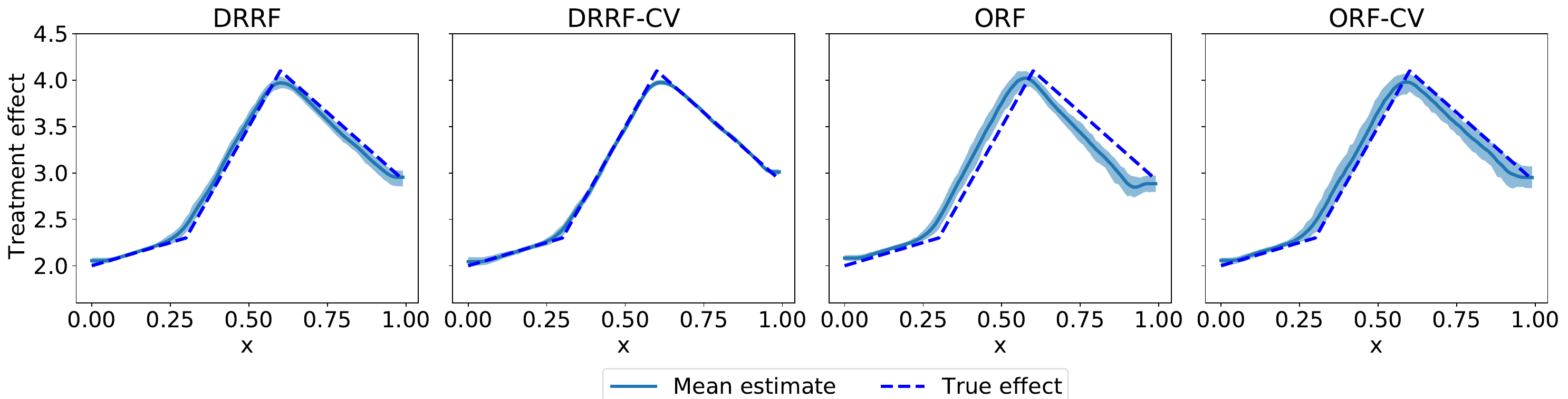} 
        \caption{$k=10$}
    \end{subfigure}
    \begin{subfigure}{\textwidth}
        \centering
        \includegraphics[width=0.8\linewidth]{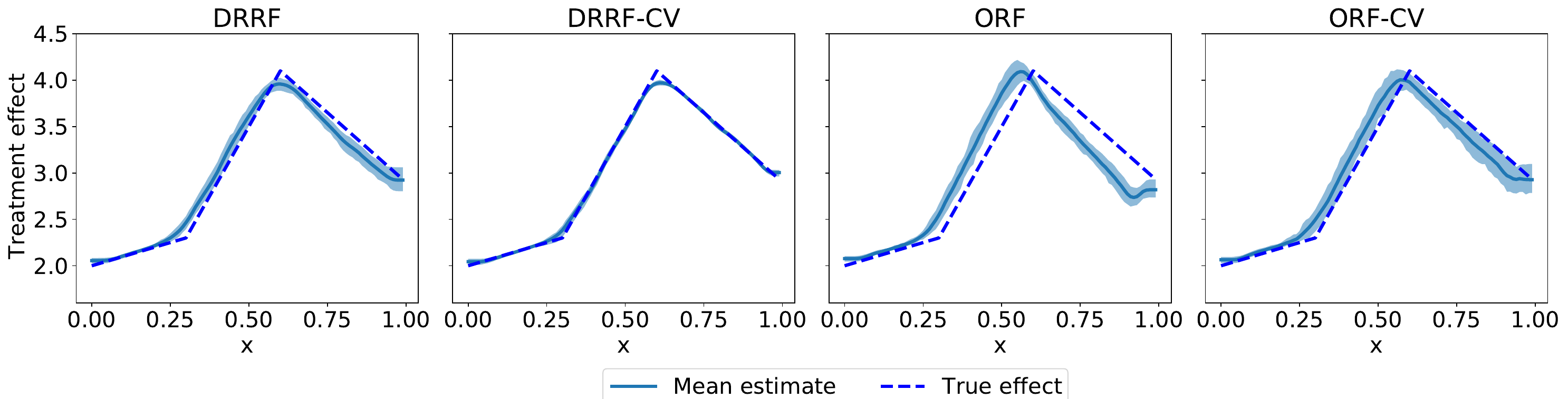} 
        \caption{$k=15$}
    \end{subfigure}
    \caption{Treatment effect estimations for Setup 2 in continuous treatment setting}
    \floatfoot{This figure shows the treatment effect estimations for Setup 2 in the continuous treatment setting. The plot shows the average estimated effects, with shaded regions representing the 5th to 95th percentile intervals of the estimated effects across 100 replications. The data is generated from the process $Y_i=\theta_0(X_i)\cdot D_i+\tilde X_i^\top \nu_0+\epsilon_i$, where $X_i\overset{i.i.d.}{\sim} U[0,1]$, $\tilde{X}_i \sim \mathcal{N}(2X_i, I_p)$, and $D_i = (\tilde{X}_i^\top \gamma_0)^3 + \eta_i$. Both $\nu_0$ and $\gamma_0$ are $k$-sparse vectors with the same support, and their non-zero coefficients are drawn i.i.d. from $U[-1,1]$. We set the number of observations to $2n=5000$ and $p=100$.}
    \label{fig:cts_2}
\end{figure}

\begin{figure}[t!]
    \centering
        \begin{subfigure}{\textwidth}
        \centering
        \includegraphics[width=0.8\linewidth]{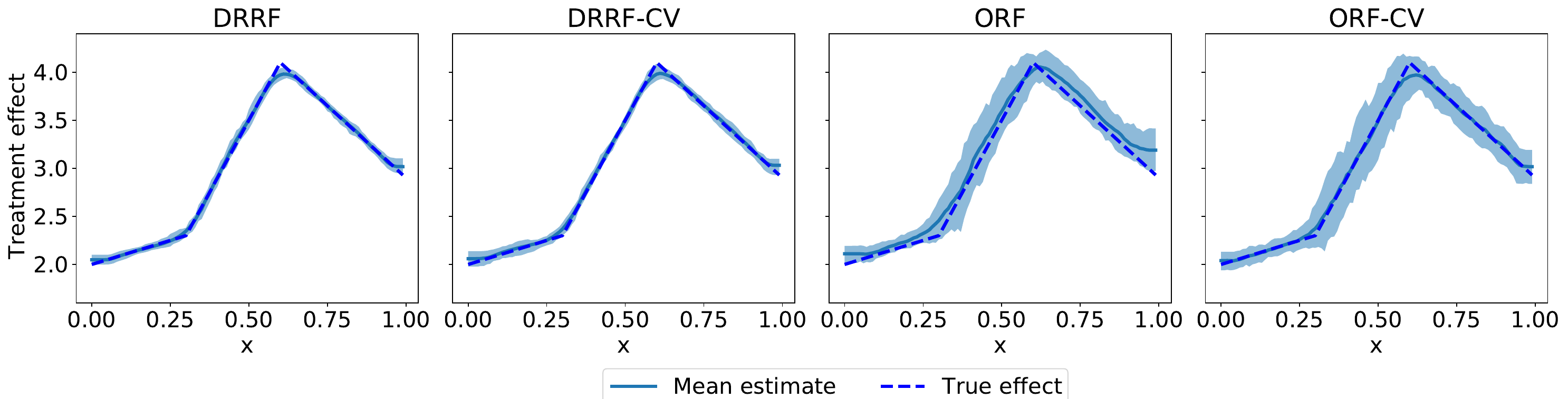} 
        \caption{$k=5$}
    \end{subfigure}
    \begin{subfigure}{\textwidth}
        \centering
        \includegraphics[width=0.8\linewidth]{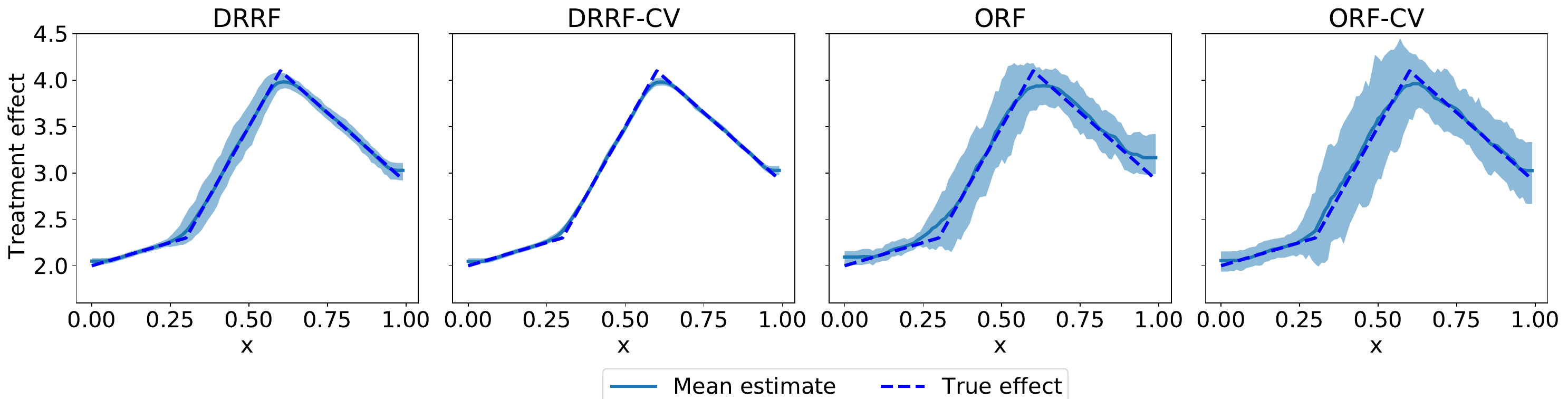} 
        \caption{$k=10$}
    \end{subfigure}
    \begin{subfigure}{\textwidth}
        \centering
        \includegraphics[width=0.8\linewidth]{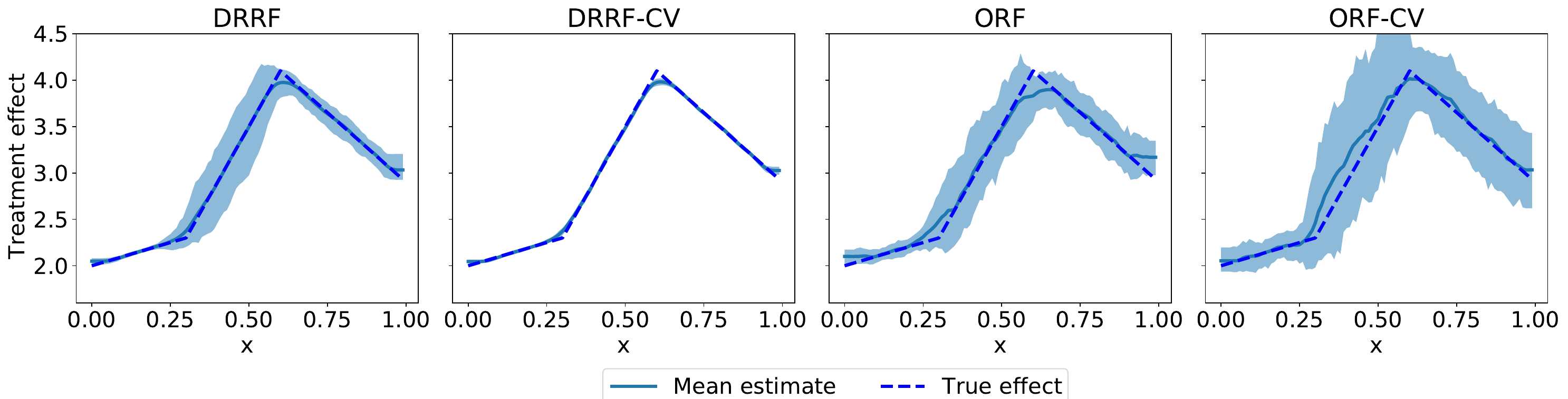} 
        \caption{$k=15$}
    \end{subfigure}
    \caption{Treatment effect estimations for Setup 3 in continuous treatment setting}
    \floatfoot{This figure shows the treatment effect estimations for Setup 3 in the continuous treatment setting. The plot shows the average estimated effects, with shaded regions representing the 5th to 95th percentile intervals of the estimated effects across 100 replications. The data is generated from the process $Y_i=\theta_0(X_i)\cdot D_i+\tilde X_i^\top \nu_0+\epsilon_i$, where $X_i\overset{i.i.d.}{\sim} U[0,1]$, $\tilde{X}_i \sim \mathcal{N}(2, I_p)$, and $D_i = (\tilde{X}_i^\top \gamma_0)^2 + \eta_i$. Both $\nu_0$ and $\gamma_0$ are $k$-sparse vectors with the same support, and their non-zero coefficients are drawn i.i.d. from $U[-1,1]$. We set the number of observations to $2n=5000$ and $p=100$.}
    \label{fig:cts_3}
\end{figure}

\end{document}